\newif\ifdraft\draftfalse
\newif\iflater\laterfalse  % for things that we can discuss later
\newif\ifpostreview\postreviewfalse
\newif\ifextended\extendedfalse
\newif\ifappendix\appendixtrue

\documentclass[acmsmall,screen]{acmart}
\usepackage{booktabs}   %% For formal tables:
\usepackage{xparse}
\usepackage{xspace}
\usepackage{xcolor}
\usepackage{relsize}
\usepackage{float}
\usepackage{cleveref}
\usepackage{syntax}
\usepackage{bm}
\usepackage{amsmath}
\usepackage{amsthm}
\usepackage{listings}
\usepackage{array}
\usepackage{mathpartir}
\usepackage{natbib}
\usepackage{graphicx}
\usepackage[inline]{enumitem}
\usepackage{rotating}
\usepackage{relsize}
\usepackage{empheq}
\usepackage{wrapfig}
\usepackage{pifont}

\setcopyright{rightsretained}
\acmPrice{}
\acmDOI{10.1145/3428233}
\acmYear{2020}
\copyrightyear{2020}
\acmSubmissionID{oopsla20main-p143-p}
\acmJournal{PACMPL}
\acmVolume{4}
\acmNumber{OOPSLA}
\acmArticle{165}
\acmMonth{11}

\definecolor{dkblue}{rgb}{0,0.1,0.5}
\definecolor{dkgreen}{rgb}{0,0.5,0}
\definecolor{dkred}{rgb}{0.7,0,0}
\definecolor{dkpurple}{rgb}{0.7,0,0.4}
\definecolor{olive}{rgb}{0.4, 0.4, 0.0}
\definecolor{teal}{rgb}{0.0,0.5,0.5}
\definecolor{azure}{rgb}{0.0, 0.4, .8}
\definecolor{edoyellow}{rgb}{0.568, 0.568, 0.008}
\definecolor{proofblue}{rgb}{0.325, 0.110, 0.702}
\definecolor{testviolet}{rgb}{0.502, 0, 0.486}
\definecolor{dkgray}{rgb}{0.5,0.5,0.5}

\long\def\comment#1{}

\newcommand{\comm}[3]{\ifdraft\textcolor{#1}{[#2: #3]}\fi}

\newcommand{\bcp}[1]{\comm{dkpurple}{BCP}{#1}}

\newcommand{\BCP}[1]{{\bf \comm{dkpurple}{BCP}{#1}}}

\newcommand{\hzh}[1]{\comm{violet}{HZH}{#1}}

\newcommand{\fuzzidp}{DPCheck\xspace}

\newtheorem{thm}{Theorem}
\newtheorem{lem}[thm]{Lemma}
\newtheorem{defn}[thm]{Definition}

\newcommand{\pr}[2]{\mathbb{P}_{#1}\!\left[#2\right]}
\newcommand{\distr}{\mathop{\bigcirc}}

\newcommand{\tyint}{\mathtt{Int}}
\newcommand{\tyreal}{\mathtt{Double}}
\newcommand{\tybool}{\mathtt{Bool}}
\newcommand{\tylist}[1]{\mathtt{[}#1\mathtt{]}}
\newcommand{\tymap}[2]{\mathtt{Map}\,#1\,#2}

\newcommand{\lap}[2]{\mathtt{lap}_{#1, #2}}
\newcommand{\aprhl}{apRHL\xspace}

\ifextended
\newenvironment{ENUM}{\begin{enumerate}}{\end{enumerate}}
\else
\newenvironment{ENUM}{\begin{enumerate*}}{\end{enumerate*}}
\fi

\newcommand{\boxedeq}[2]{\begin{empheq}[box={\fboxsep=6pt\fbox}]{align}\label{#1}#2\end{empheq}}

\newcommand{\aprhlstmt}[6]{{\color{dkgray}\mathlarger{\vdash}}\,
  #1\mathrel{{\color{dkgray}\sim}_{(#2, #3)}} #4
  \mathrel{\color{dkgray}\mbox{\Large :}} #5 \mathrel{\color{dkgray}\Rightarrow} #6}

\newcommand{\cmark}{\ding{51}}%
\newcommand{\xmark}{\ding{55}}%

\author{Hengchu Zhang}
\affiliation{
  \institution{University of Pennsylvania}            %% \institution is required
  \country{USA}
}
\email{hengchu@seas.upenn.edu}          %% \email is recommended

\author{Edo Roth}
\affiliation{
  \institution{University of Pennsylvania}            %% \institution is required
  \country{USA}
}
\email{edoroth@seas.upenn.edu}          %% \email is recommended

\author{Andreas Haeberlen}
\affiliation{
  \institution{University of Pennsylvania}            %% \institution is required
  \country{USA}
}
\email{ahae@cis.upenn.edu}          %% \email is recommended

\author{Benjamin C. Pierce}
\affiliation{
  \institution{University of Pennsylvania}            %% \institution is required
  \country{USA}
}
\email{bcpierce@cis.upenn.edu}          %% \email is recommended

\author{Aaron Roth}
\affiliation{
  \institution{University of Pennsylvania}            %% \institution is required
  \country{USA}
}
\email{aaroth@cis.upenn.edu}          %% \email is recommended

\begin{document}

\title{Testing Differential Privacy with Dual Interpreters}% with Instrumented and Symbolic Execution}


\begin{abstract}
Applying differential privacy at scale requires convenient ways to check that
programs
computing with sensitive data appropriately preserve privacy.  We propose here a
fully automated framework for {\em testing} differential privacy, adapting a well-known
``pointwise'' technique from informal proofs of differential privacy. Our
framework, called \fuzzidp\bcp{I wonder whether we can think of a nicer /
  more interesting name; this one is OK, but a bit clunky and a bit
  misleading (because what we're doing is not really fuzz testing)},
requires no programmer annotations, handles all previously verified or
  tested
algorithms, and is the first fully automated framework to distinguish correct and buggy implementations of PrivTree, a
probabilistically terminating algorithm that has not previously been
mechanically checked.

We analyze the probability of \fuzzidp mistakenly accepting a non-private
program and prove that, theoretically, the probability of false acceptance can
be made exponentially small by suitable choice of test size.

We demonstrate \fuzzidp's utility empirically by implementing all benchmark
algorithms from prior work on mechanical verification of differential privacy,
plus several others and their incorrect variants, and show \fuzzidp accepts the
correct implementations and rejects the incorrect variants.

We also demonstrate how \fuzzidp can be deployed in a practical workflow to test
differentially privacy for the 2020 US Census Disclosure Avoidance System (DAS).
\end{abstract}

\maketitle

\section{Introduction}
\label{sec:intro}
Differential privacy, the highest standard for privacy-preserving data analysis,
%\ar{The gold standard was a famously bad monetary policy. Should we say that DP is the fiat currency of privacy preserving data analysis?}\hzh{I think the analogy with ``standard'' makes more sense than with ``currency''... is there something like a ``fiat standard''?}\ar{This was mostly a joke, but ``gold standard'' is a weird way of saying something is good.}\hzh{Maybe we can call it the ``bitcoin'' of privacy preserving data analysis... but I'll try something like ``highest standard'' for now?}
is
 seeing increasing real-world
 deployment~\cite[etc.]{MSDP,AppleDP,CensusDP}. Differential privacy offers
 precise guarantees, is robust to arbitrary post-processing, and gives a
 quantitative estimate of privacy loss. However, differentially private
 algorithms often require subtle reasoning for their proofs of privacy. Even
 experts get these proofs wrong
\cite{Lyu:2017:USV:3055330.3055331}.

% The subtlety of differential privacy proofs has motivated research on
% systems that aid implementations of differentially private algorithms.
% Differential privacy is a relation on the output distribution between randomized
% programs, and is challenging to check automatically. Nevertheless, a recent line
% of work
A number of recent efforts
have focused on partly or fully automating the process of
certifying differential privacy for sophisticated
algorithms such as
the sparse vector technique \cite{Dwork:2014:AFD:2693052.2693053}. For
example, Zhang and
Kifer's LightDP \cite{Zhang:2017:LTA:3093333.3009884} uses a lightweight
dependent type system along with some
programmer annotations; \citet{Ding:2018:DVD:3243734.3243818} propose a fully
automated statistical testing framework that attempts to disprove
differential privacy of queries using statistical
evidence; \citet{Albarghouthi:2017:SCP:3177123.3158146} demonstrate a completely
automated proof synthesis system
with a specialized program logic;
%\bcp{are ``pointwise'' and ``coupling'' the
%same thing?  We should use consistent terminology if so.}
%\hzh{no they are kind of orthogonal: pointwise can be formalized
%in other proof systems too, not necessarily coupling methods}
and
\citet{Wang:2019:PDP:3314221.3314619} use a
proof technique called ``Shadow Execution'' to improve upon LightDP and
further reduce the amount of programmer annotations required and increase the
performance of verification.

We propose a novel framework, called \fuzzidp, that requires no annotations,
handles key benchmark algorithms from previous work, and moreover can accept
correct implementation, and reject faulty variants of the
PrivTree \cite{Zhang:2016:PDP:2882903.2882928} algorithm, a challenging algorithm for automatic differential privacy verification acknowledged by authors of LightDP~\cite{Zhang:2017:LTA:3009837.3009884}.
%\bcp{I wish we could say
%``algorithms such as PrivTree...'' (i.e., describe a class of algororithms we
%can now handle, rather than a single one) }
%\hzh{I wish so too... but so far the trickier algorithms seem to have gone beyond FuzzDP's ability.}\ar{Can we at least point to privtree being an acknowledged challenge problem?}
PrivTree was beyond the scope of previous work, owing to
a probabilistic main loop that terminates eventually with probability $1$ but
is not guaranteed to terminate in any bounded number of iterations.
Our key insight is that we can combine information from {\em instrumented}
and {\em symbolic} executions of a
program to construct privacy proofs for specific executions,
%\bcp{``check whether
%potential privacy proofs exist'' is confusing---what is a ``potential proof''?}
%\bcp{Still confusing.  Why not something like ``construct privacy proofs for
%  specific executions'' (and then note that if we successfully do this
%  enough times we get a random-DP-style probabilistic guarantee about its
%  privacy)}
then combine these proofs from a large number of
executions to give a statistical guarantee of random differential
privacy~\cite{Hall2013}. The
symbolic interpreter uses information gathered by the instrumented interpreter to
build a simple static analysis of the program's privacy properties,
making automated
detection of privacy leaks, even for challenging algorithms, feasible.

%HZH: I feel this whole paragraph does not add much...
%\fuzzidp draws from both statistical testing and proof synthesis for differential privacy.
%\fuzzidp adapts the underlying proof system from \cite{Albarghouthi:2017:SCP:3177123.3158146} into a testing strategy.
%%\bcp{Don't
%%  understand this sentence:}
%Compared to StatDP's statistical testing
%approach \cite{Ding:2018:DVD:3243734.3243818}, \fuzzidp provides more
%information because it directly checks the existence of certain privacy proofs
%for the program under test. Furthermore, \fuzzidp's testing strategy enables it
%to check algorithms that were beyond the reach of existing work---in particular,
%it can automatically check the privacy of
%PrivTree, which is cited as a challenging
%case in \cite{Zhang:2017:LTA:3009837.3009884}.

%\bcp{The grammatical tense switches from future to present halfway through the
%paragraph.  Just use present tense.}  \bcp{This opening is confusing: \fuzzidp
%itself is arguably part of the first contribution, and the material in section
%5 is definitely part of it!}  HZH: the embedded design of FuzzDP isn't
%original though... it's adapted from Svenningsson and Axelsson's paper.
Following a short review of differential privacy in \Cref{sec:review} and an
overview of \fuzzidp's syntax and semantics
in \Cref{sec:fuzzidp-types-and-syntax} and \Cref{sec:semantics}, we offer these
main contributions:
\begin{enumerate}
\item We present a testing strategy for differential privacy adapted from the
{\em pointwise proof technique} (\Cref{sec:testing-differential-privacy}).
\item We prove that the testing strategy always correctly accepts a class of
well-behaved differentially private programs, and prove that, in principle, this strategy's probability of
incorrectly accepting ill-behaved (non-private) programs decreases
%\bcp{What if we worded
%this as ``can in principle be made...''?  I think that would signal that we
%have not done it in practice.}
exponentially as test size increases, for a class of ill-behaved programs defined
%\bcp{I think this went the wrong direction; we mean
%  ``prove... with'', not ``can be made exponentially small with...'', but it
%reads the opposite way}
in the framework of random differential privacy
%\ar{Strange phrasing, because we aren't using any ``theory'' from random differential privacy, are we? We mean here that we define ``ill behaved'' to be ``not even random differentially private'' right?}
(\Cref{sec:random-differential-privacy}). We then give an overview of the
testing strategy implementation for \fuzzidp in Haskell (\Cref{subsec:implementation}).
\item We demonstrate the effectiveness of the testing strategy by
showing that it can detect non-private variants with common programming
mistakes, and published mistakes made by experts on the sophisticated benchmark
algorithms\ifextended{} (ReportNoisyMax, ReportNoisyMaxWithGap, SparseVector,
SparseVectorWithGap, PrefixSum, SmartSum, and PrivTree) and simpler algorithms
(NoisySum, NoisyCount, and NoisyMean)\fi, and that it accepts correct
implementations of all these algorithms. In particular, \fuzzidp is the first
automated framework that can distinguish correct and incorrect variants of
PrivTree (\Cref{subsec:benchmark}). These benchmark algorithms' differential
privacy proofs cover a wide range of complexity, demonstrating \fuzzidp can
analyze both simple and sophisticated differentially private programs.
\item We present a practical workflow that uses \fuzzidp to re-implement and test the
core differential privacy mechanisms in the Disclosure Avoidance System
(DAS)~\cite{10.12688/gatesopenres.13089.1} designed for 2020 US Census; we also
show statistical evidence that our re-implemented core mechanism behaves the
same as the unmodified DAS (\Cref{subsec:das}).
\item We implement \fuzzidp as an embedded language in Haskell
%\bcp{The way this is worded now, it sounds like a novel
%  contribution again, so I'm moving it back here.}%
and discuss a type-driven optimization adapted
from \citet{Torlak:2014:LSV:2666356.2594340} to speed up symbolic execution,
which improves testing time for some our benchmark algorithms
(\Cref{sec:optimizations}).
\end{enumerate}
\Cref{sec:limitation} enumerates some
limitations, and \Cref{sec:related-work,sec:conclusion} discuss
related and future work.
%
%Our primary contributions are as follows: \bcp{It's a bit strange for
%  contributions to come after the paper outline.  After is better, but even
%  better is to merge them.}
%
%\begin{enumerate}
%\item We demonstrate the versatility our system on a suite of well-known
%benchmark algorithms: ReportNoisyMax, SparseVector,
%SparseVectorWithGap, and their known incorrect variants; in particular, \fuzzidp
%is the first automated framework to check
%PrivTree \cite{Zhang:2016:PDP:2882903.2882928}, a challenging algorithm that
%builds spatial decomposition trees. \bcp{This item is a bit confusing
%  because it's mixing positive and negative results: we are able to {\em
%    ``prove''} (with an exponentially small probability of being wrong) that
%some algorithms are private, and we are able to find {\em counterexamples}
%for others.  Would be clearer to present these claims separately.}
%\item
%\item We discuss the connection between our framework and
%rigorous proofs of differential privacy written in \aprhl. \hzh{I think we don't have enough space to spell out everything for this to be
%a contribution.} \bcp{But maybe still worth mentioning in the outline of the
%rest of the paper?}
%\end{enumerate}

\section{Background}
\label{sec:review}
A {\em discrete distribution} over values of type $\tau$ is a function of type
$\tau \mapsto [0, 1]$ mapping each value in $\tau$ to an associated
probability. We write $\distr\,\tau$ for the set of discrete distributions over values in
type~$\tau$. An {\em event} $E$ is a subset of $\tau$. The {\em support} of a discrete
distribution $\mu$ is the subset of $\tau$ whose values have non-zero
probability: $\mathtt{supp}(\mu) = \{x \in \tau
\,|\, \mu(x) > 0\}$.
\begin{defn}
\label{defn:sub-distribution}
Let $\mu :: \distr\,\tau$ be a discrete distribution.
We call $\mu$ a {\em sub-distribution}
if the sum of probabilities over its support is at most $1$:
%\bcp{Does this
%  definition actually work for continuous distributions??}
%\hzh{this is only for discrete distributions}
\ifextended$$\else$\fi
  \sum_{v\in \mathtt{supp}\,\mu} \mu(v) \leq 1
\ifextended$$ \else$. \fi
%
%\bcp{Do we ever use this?}
%HZH: I plan on using it in the the to-be-improved discussion about discrete laplace/two-sided geometric distribution
%HZH: not used anymore
%Similarly, we call $\mu$ a {\em proper distribution} if the sum of probabilities over its
%support is exactly $1$\ifextended:
%$$\sum_{v\in \mathtt{supp}\,\mu} \mu(v) = 1$$
%\else. \fi
\end{defn}

Sub-distributions are useful for describing the semantics of randomized
programs because they naturally model non-termination through the ``missing
probability.''  In what follows, we write just ``distribution'' to mean
sub-distribution.%\bcp{Do we ever write ``distribution'' to mean a full
%  distribution?}

Differential privacy is a relational property of randomized programs.
Informally, a program is differentially private if it produces similar
distributions when run on similar inputs. The exact similarity relation on
inputs depends on what private information we care about protecting.
For example, a program $f$ may be counting the number of patients diagnosed with
some disease in a medical database; to conform to regulations, we
must not leak the diagnosis of any particular patient. More
precisely, the distribution of outputs should be nearly the same if the
diagnosis of {any} single patient changes in the input
database. For this example, an appropriate similarity relation on inputs is ``at
most one patient's data may be different between the two input databases,''
or, more generally:
\begin{defn}
\label{defn:database-distance}
Two multisets have {\em database distance} $k$ if at most $k$ items must be
added or removed to make the two contain exactly the same items.
\end{defn}

Another common similarity relation is the
$L1$-distance between two vectors of numbers.
\begin{defn}
\label{defn:l1-distance}
The {\em $L1$ distance} between vectors $x_1$ and $x_2$ is the sum of the
coordinate-wise distances between corresponding elements of the two vectors.
%:
%\bcp{Don't think the notation is ever used.}
%$$\norm{x_1 - x_2}_1 = \sum_{i=0}^{n-1} |x_1[i] - x_2[i]|$$
\end{defn}

Finally, some algorithms' notion of similar inputs is vectors with bounded
coordinate-wise distance.
\begin{defn}
\label{defn:coordinate-wise-distance-relation}
Vectors $x_1$ and $x_2$ have {\em coordinate-wise distance
$k$} if  $\left|x_1[i] -
x_2[i]\right|$ is bounded by $k$ for each coordinate~$i$.
\end{defn}

Here is the fundamental definition of differential privacy:

\begin{defn}
\label{defn:epsilon-delta-differential-privacy}
A randomized program $f :: \tau \mapsto
\distr \sigma$ is {\em $(\epsilon, \delta)$-differentially
private} if, for all similar inputs $(x_1,
x_2) \in \tau \times \tau$, the probability of any event
$E \subseteq \sigma$ satisfies the inequality
\begin{equation}
\pr{f(x_1)}{E} \leq e^\epsilon \pr{f(x_2)}{E} + \delta \label{eq:1}.
\end{equation}
\end{defn}

If the support
%\bcp{why not just say ``domain''?}\hzh{strictly speaking, the
%domain could be uncountable. as long as the support is countable,
%then the following holds}
of the probability distributions is countable, we can simplify the
definition of $(\epsilon, 0)$-differential privacy using a pointwise inequality
on the probability difference:
\begin{defn}
\label{defn:epsilon-differential-privacy-for-discrete-domains}
A randomized program
%\bcp{Are we using one colon or two for type membership?
%(I think I asked about this before...)}
%\hzh{let's stick with two colons to be consistent with haskell code listing}
$f :: \tau \mapsto \distr \sigma$, for
some countable domain
$\sigma$, is $(\epsilon, 0)$-{\em differentially private} if, for all similar
inputs $(x_1, x_2) \in \tau \times \tau$, the probability of any singleton event
$v \in \sigma$ satisfies
$$\pr{f(x_1)}{\{v\}} \leq e^\epsilon \pr{f(x_2)}{\{v\}}.$$
\end{defn}

The $\epsilon$ and $\delta$ in the definition of differential privacy are
``privacy parameters.'' We can interpret them as a quantitative
measure of how much privacy is lost when a sample is observed from the output
distribution. As $\epsilon$ increases, the multiplicative bound on the
difference in probabilities of output events becomes looser, increasing
an attacker's confidence in
% can confidently\bcp{added ``confidently,'' but I'm not
%  certain I've got the order of quantifiers right: should it be ``increasing
%the attacker's {\em confidence}...''?}
distinguishing two executions of $f$ on
similar inputs.
Experts recommend picking small $\epsilon$ values
(e.g., $1.0$) for meaningful privacy protection~\cite{Hsu:2014:DPE:2708449.2708711}. %\bcp{$\epsilon$ is not usually chosen when
%  programs are {\em designed}, is it?}
On the other hand, $\delta$
bounds the probability of ``catastrophic failure''---failure to provide any
privacy at all. It should generally be very small.%\bcp{Can we say, ``It should generally be very small''?}
%\bcp{why do we need to discuss how it should be chosen?
%  we're only using it in the proof, right? i.e., no one is going to be
%  choosing it.}it should
%be chosen smaller than inverse of the input size, since otherwise a program can
%trivially satisfy $(0, \delta)$-differential privacy by randomly emitting an
%element from the input.
%\bcp{$\delta$ only has to be a little bit smaller than
%$1/n$ to avoid this, no?}

\fuzzidp only guarantees to accept %\bcp{weird phrasing: is this a goal, or a limitation?}
programs that achieve
$(\epsilon, 0)$-differential privacy. %\ar{This isn't true, right? We mean it is only guaranteed to accept these?}.
However, nonzero $\delta$s will play a
role in our analysis of false negatives---tests in which \fuzzidp fails to
detect a non-$(\epsilon, 0)$-differentially private program.

An important tool for writing differentially private algorithms is the Laplace
distribution. It is commonly defined as a continuous
distribution, but rigorous proofs of differential privacy using continuous
distributions require sophisticated measure theory \cite{8785668}. To simplify
the required foundations, we follow
previous work on program semantics and differential
privacy \cite{Albarghouthi:2017:SCP:3177123.3158146,
DBLP:journals/corr/abs-1710-09951, Wang:2019:PDP:3314221.3314619,
Zhang:2017:LTA:3009837.3009884, Reed:2010:DMT:1863543.1863568} and assume a
discretized, countable support over the reals for all representable numbers. We
write $\omega$ for the constant
%\BCP{is it constant in practice??  I
%  thought the
%  gap between consecutive large floating point numbers was larger than the
%  gap between
%  consecutive small floats...?}
gap\footnote{We present \fuzzidp and its properties using this idealized set of representable reals, but our implementation relies on floating point numbers. This discrepancy and its impact on testing are discussed in \Cref{sec:limitation}.} between consecutive values---the
\emph{granularity} of the discretized domain. In this work, we assume all
real values
are drawn from this discretized domain with granularity $\omega$.
%\bcp{Given this setup, I was expecting the next definition to be some kind of
%discretized Laplace distribution.  But it seems to be the standard
%one...}
%\hzh{the difference is that the probability is proportional to, instead of
%equal to the equaiton on RHS. the usual continuous Laplace distribution is
%defined in terms of a probability density function that's equal to the
%RHS. Here, we assume a discretized Laplace with the given probability mass
%function.}
%\bcp{There should be some discussion of this difference!}
%\begin{defn}[\cite{Ghosh:2009:UUP:1536414.1536464}]

\label{defn:laplace-distribution}
%\bcp{Why make this a definition?}%
The {\em  discretized Laplace distribution} is formally a {\em two-sided
geometric distribution}~\cite{Ghosh:2009:UUP:1536414.1536464}. It
is parameterized by a center $c$ and a parameter
$\alpha \in [0, 1]$. \citet{Ghosh:2009:UUP:1536414.1536464} show the two-sided
geometric distribution shares the important privacy properties of the continuous
Laplace distribution. The continuous Laplace
distribution is parameterized by a center $c$ and a width parameter $w$ that
controls how centered the distribution is around
$c$. \citet{Ghosh:2009:UUP:1536414.1536464} also show a straightforward
translation between the two-sided geometric distribution parameter $\alpha$ and
the corresponding parameter $w$ for an equivalent discretized Laplace
distribution. We will exclusively use the width parameter $w$ to parameterize discrete
Laplace distributions in this work.
%\bcp{Looks
%  good, but: (a) it's using up a lot of space and (b) the little rectangles
%  are hard to see (in part because they look a bit blurry).}%
%Its probability on a sample $x$ is proportional to \bcp{I wonder if there's any
%point in showing the formula---do we need it later? (And even if we do, maybe
%we can just introduce it right there?)  IMO, a little picture would be more
%informative.}  $$\mathtt{lap}_{c,
%w}(x) \propto \frac{1}{2w}\exp{\left(-\frac{|x - c|}{w}\right)}.$$
%\end{defn}

\begingroup
\setlength{\columnsep}{2pt}%
\setlength{\intextsep}{0pt}%
\begin{wrapfigure}{r}{0.3\linewidth}
\centering
\includegraphics[width=0.9\linewidth]{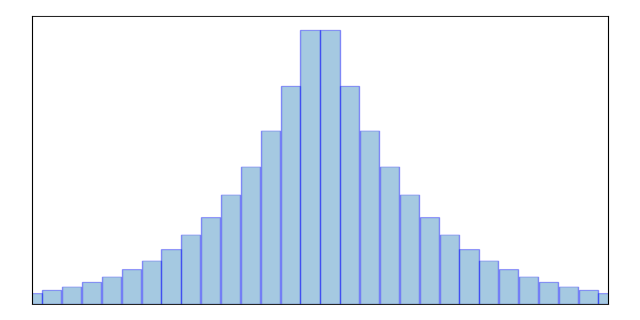}
\end{wrapfigure}
Each rectangle's area in the graph to the right represents the
probability assigned to the
point at the center of the base of the rectangle; each rectangle has width
exactly $\omega$---the granularity of the set of representable numbers. We write $\mathtt{lap}_{c, w}$ for the discretized Laplace
distribution with center $c$ and noise width $w$.
%\bcp{why do we switch back and
%  force between ``$\omega$'' and ``w''?}\hzh{$w$ here is a different parameter for {\em the laplace distrib%ution}, not for the set of representable numbers. I tweaked the wording to make it less confusing.}
%
%the empty line below is crucial for wrapfigure to work.

\endgroup

\section{Types and Syntax}
\label{sec:fuzzidp-types-and-syntax}
\fuzzidp is a testing framework built around a probabilistic programming
language embedded in the functional language Haskell. We will refer to the
embedded language
itself as \fuzzidp and to the rest of our system as the
``testing framework.''

The \fuzzidp language offers a simple, purely functional,
%\bcp{Didn't the
%  ``i'' in ``Fuzz'' used to stand for ``imperative''?}\hzh{indeed... maybe we should change the name to FuzzDP?},
notation for
differentially private programming.  \fuzzidp provides base types
%\bcp{is the ``such as'' qualification important?}
$\tybool$, $\tyint$, and $\tyreal$, as well as three container
types: tuples $(\tau,
\sigma)$, lists $\tylist{\tau}$, and maps
$\tymap{\tau}{\sigma}$. Finally, \fuzzidp supports probabilistic programming
through a distribution monad $\distr$: the type $\distr \tau$ represents a
(sub-)distribution over values of type $\tau$.  Unlike a number of previously proposed functional languages for differential
privacy~\cite[etc.]{Reed:2010:DMT:1863543.1863568,Gaboardi:2013:LDT:2480359.2429113}, \fuzzidp's
type system does {\em not} track privacy: this is the job of the testing framework.

We embed \fuzzidp inside Haskell, using
methods developed by \citet{Svenningsson:2012:CDS:3080050.3080054,SVENNINGSSON2015143} for the
Feldspar language \cite{5558637}.
%\bcp{There seem to be some
%  more recent citations; maybe we should include them too?  E.g. ``Combining
%  deep and shallow embedding of domain-specific languages'', from 2015}.
Their key insight is a technique for combining ``deep'' and ``shallow''
representations of programs, where a deep representation for a program is an abstract
syntax tree, while a shallow representation maps language constructs directly to
their semantics. Here, the deeply represented parts of the language can be used
by \fuzzidp's symbolic interpreter for static analysis, while the shallowly
represented parts save engineering on things like surface syntax, standard
libraries, and compilation by borrowing from the host language.

\fuzzidp's deep representation in Haskell uses values of
an indexed datatype \lstinline|Expr|. For
example, a Haskell term of type
$\mathtt{Expr}\,\tybool$ represents a \fuzzidp program that yields values of
type $\tybool$ when evaluated. The type index ($\tybool$) allows us to borrow Haskell's typechecker to rule out
ill-formed programs statically.

\fuzzidp's most important syntactic forms are:
\begin{lstlisting}
data Expr a where
  Literal :: a -> Expr a
  Add     :: Expr Double -> Expr Double -> Expr Double
  Lt      :: Expr Double -> Expr Double -> Expr Bool
  If      :: Expr Bool -> Expr (&$\distr$& a) -> Expr (&$\distr$& a) -> Expr (&$\distr$& a)
  Laplace :: Expr Double -> Double -> Expr (&$\distr$& a)
  Return  :: Expr a -> Expr (&$\distr$& a)
  Bind    :: Expr (&$\distr$& a) -> (Expr a -> Expr (&$\distr$& b)) -> Expr (&$\distr$& b)
\end{lstlisting}
These are Haskell constructors allowing us to build \fuzzidp programs that
perform arithmetic using \lstinline|Add|, compare numeric values using
\lstinline|Lt|,
branch on boolean conditions using \lstinline|If|, sample from Laplace
distributions using \lstinline|Laplace|,
%repeat probabilistic computations
%through \lstinline|Loop|\bcp{I don't understand the type of Loop},
sequence probabilistic computations
using \lstinline|Bind|, and create point mass distributions
using \lstinline|Return|.\footnote{The width parameter to
the \lstinline|Laplace| constructor has type $\tyreal$ instead of
$\mathtt{Expr}\,\tyreal$. This implies the width parameter must be a statically
chosen constant. This restriction simplifies the testing framework
implementation, and it does not rule out any algorithms in our evaluation.}
The \lstinline|Bind| and \lstinline|Return| constructors endow \fuzzidp with
monadic structure~\cite{39155}, allowing computations that return distributions
to be coded in a natural style.
%\bcp{Non-functional-programmers (i.e., much
%  of the SPLASH audience) will need more intuition than that, either here or
%  in the description of the example.  Also, it's a shame the example doesn't
%use Loop.}
Additionally,
the purity of Haskell automatically rules out programs with potentially
non-private side-effects (such as using the current system time to calculate an
argument to \lstinline|Literal|), because the
types of the constructors only allow effect manipulation related to probability
distributions.\footnote{In the current implementation, it is still possible to
escape \lstinline|Expr|'s restrictions by using unsafe language features to
subvert Haskell's type system.  We could fix this by requiring \fuzzidp
programmers to only use Safe Haskell \cite{Terei:2012:SH:2364506.2364524},
ruling out such subversions completely.  In this paper, we assume the programmer
is not adversarial and only wants to test programs that she genuinely believes
are differentially private.}

%\bcp{I wonder if it would work better to move the next three paragraphs to
%  after the example...}
%
As an example, we consider implementing the ReportNoisyMax algorithm in \fuzzidp. The ReportNoisyMax algorithm \ifextended adds noise to each value in an input list, and
\fi returns the
index of the largest noisy value \ifextended in the list\else\fi. Assuming an input list of \lstinline|[1, 2, 10]| to ReportNoisyMax, and the noise values are $0.8, -1.2$ and $-0.9$, ReportNoisyMax returns the index of the largest value from the noised list \lstinline|[1.8, 0.8, 9.1]|. We can define ReportNoisyMax
%\bcp{some readers
%  may not know what this is---we're referring to it like they should. How
%  about moving the explanation a bit later to right here, and expanding it a
%bit---e.g., with a concrete example or two?}
with
the following
Haskell syntax and a Haskell library combinator \lstinline|mapM| that
applies a function uniformly over a list:%
\label{rnm}
\begin{lstlisting}
rnm :: [Expr Double] -> Expr (&$\distr$& Int)
rnm (x:xs) = do
  xNoised  <- lap x 1.0
  xsNoised <- mapM (\y. lap y 1.0) xs
  rnmAux xsNoised 0 0 xNoised
rnm [] = error "rnm received empty input"

rnmAux :: [Expr Double] -> Expr Int -> Expr Int -> Expr Double
       -> Expr (&$\distr$& Int)
rnmAux [] _ maxIdx _ = return maxIdx
rnmAux (x:xs) lastIdx maxIdx currMax = do
  let thisIdx = lastIdx + 1
  if (x > currMax)
     (rnmAux xs thisIdx thisIdx x)
     (rnmAux xs thisIdx maxIdx  currMax)
\end{lstlisting}

The left pointing arrow \lstinline|<-| is a syntactic sugar of the \lstinline|Bind| constructor for sequencing probabilistic computations. Intuitively, the syntax \lstinline|x <- m| represents a program fragment that runs the probabilistic computation \lstinline|m|, giving the result of that computation a name \lstinline|x|, and allows \lstinline|x| to be used in following computations.

Conceptually, executing \fuzzidp programs through Haskell is a two-stage
process: Haskell itself becomes a host language for
constructing embedded \fuzzidp programs, which are then executed using an {\tt eval}
function (see \Cref{ap:eval-implementation}
%\bcp{We've checked that SPLASH
%  allows subnissions with appendices?} \hzh{Yes.}
  ). This arrangement allows programmers to use convenient
Haskell syntax and libraries when writing \fuzzidp programs.

An unusual aspect of \fuzzidp's syntax design is that we rely on recursion in Haskell for---in effect---generating
iterative \fuzzidp code.
%The \lstinline|Loop| constructor produces smaller code
%representations which are better suited for extraction of \fuzzidp code\bcp{???}, while
%iteration through host-level recursion presents the embedded \fuzzidp code in a
%more functional style. In our running examples, we prefer the iteration through
%host-level recursion style unless extraction is required.\bcp{Readers don't know what
%  we mean by extraction here.  (And I'm not actually sure I understand the point
%  we're making.)}
%
%\BCP{I find everything to do with Loop here pretty confusing.  I wonder if
%  we should actually omit it from the syntax and just explain recursion,
%  perhaps with a footnote or forward reference about issues with extraction
%  and how we address them with Loop.}
%
This design relies critically on Haskell's support for lazy
evaluation.
Iteration through host-level recursion, if implemented in a strictly evaluated
host language, would restrict the kinds of algorithms that can be represented,
since a strict host-language would fully construct the syntax tree before sending
it off to an interpreter. This implies all loops built with strict host-language
recursion would be fully unrolled immediately. Some programs may have unbounded
loops in some control flow paths, and such programs will cause divergence under
this scheme. Fortunately, since Haskell is lazily evaluated, its runtime does
not construct a value until the value's contents are required. This means the
Haskell code that constructs \fuzzidp syntax does not run until that syntax is
required by the interpreter. In our implementation, we rely on Haskell's lazy
evaluation to automatically interleave \fuzzidp code execution and (potentially
infinite) loop unrolling.\footnote{Although the ordinary \fuzzidp interpreter
handles infinite syntax trees thanks to laziness, this doesn't help the {\em
symbolic} interpreter, because it tries to explore all paths. However, we can
use auxiliary information to stop potentially infinite symbolic execution as
soon as no further control flow path exploration is
necessary. See \Cref{sec:limitation}.}
%
\iffalse
\iflater
\bcp{I don't think this explanation quite nails
it---isn't the point simply that lazy evaluation allows us to represent infinite
syntax trees without diverging.  (But actually, thinking about symbolic
execution, I realize I don't understand why even lazy evaluation won't diverge
as soon as we try to do anything interesting with infinite programs, like
calculating path conditions for all possible outputs...  Is there something more
we can say here?)} \hzh{Our symbolic execution currently does not deal with
infinite syntax trees, and instead uses a hack to stop unrolling when we realize
no more symbolic outputs can be matched with instrumented outputs. This hack is
only used in the evaluation of PrivTree, since it's the only algorithm that
cannot be fully statically unrolled. This is discussed in the last paragraph of
limitation. I'm not sure saying more about this here is helpful. However, the
instrumented interpreter can properly deal with infinite syntax trees.}
\bcp{Maybe worth at least putting a footnote here briefly acknowledging the
  issue and giving a pointer to the later discussion.}
\fi
\fi

%
When the function \lstinline|rnm| is applied
to an input list $\mathtt{xs}$, we obtain
a closed \fuzzidp program with type $\mathtt{Expr}\,(\distr\mathtt{Int})$. This program
contains data that represents \fuzzidp commands to be
interpreted. The initial commands are Laplace sampling commands that add noise
to the input list. These commands are generated with
%\bcp{I
%  don't understand what ``These commands are encoded through the function'' means...}
$\lambda\, \mathtt{y}.\, \mathtt{lap}\, \mathtt{y}\, 1.0$, which builds
the \lstinline|Laplace| nodes in a \fuzzidp program's syntax tree,
and \lstinline|mapM|, which applies this function for all items in a list. We
then loop over the noised data, keeping track of the maximum
value and its index seen so far through \lstinline|If| commands, and return the
index at the end with the helper function \lstinline|rnmAux|.
\ifextended
This recursive
loop is unrolled one iteration at a time during the interpretation of
the \fuzzidp program instead of all at once, thanks to Haskell's lazy
evaluation.
\fi
%This design of FuzzDP allows programmers to construct \fuzzidp programs while
%reusing Haskell's syntax and vast library ecosystem.
%\bcp{But not just
%experiment: The haskell embedding also gives \fuzzidp programs access to a ton
%of Haskell libraries, no?}

\section{Semantics}
\label{sec:semantics}
We can straightforwardly encode \fuzzidp's evaluation semantics using a
function
$\mathtt{eval} :: \mathtt{Expr}\, \mathtt{r} \rightarrow \mathtt{r}$ (see
\Cref{ap:eval-implementation} for details). Recall that a
differentially private program whose syntax tree has type
$\mathtt{Expr}\,(\distr\, a)$ is necessarily
probabilistic. Calling \lstinline|eval| on such a program results in a
distribution (value) of type $\distr\, a$, and we can sample from this
distribution to acquire a concrete output value.

Besides this evaluation semantics, the testing framework also uses symbolic
execution \cite{King:1976:SEP:360248.360252}
of \fuzzidp programs.  During symbolic execution, concrete samples from Laplace
distributions are replaced by symbols,
allowing us to explore all possible control flows throughout a \fuzzidp program,
even when branch conditions depend on sampled values. The symbolic execution
process records the branch conditions along each control flow path; we call
these recorded boolean conditions \emph{path conditions}. As an example, if we
run \lstinline|rnm| on an input list of length two, then the symbolic execution
process will create two symbolic values $s_0$ and $s_1$, one for each noised
sample
value. On one path, we explore under the assumption that the branch condition
$s_0 > s_1$ evaluates to \lstinline|True|, and this results, say, in the
output value $0$; on the other path, we explore under the assumption that
$s_0 > s_1$
evaluates to \lstinline|False|, and this results in the output value $1$. So, in
this example, symbolically executing \lstinline|rnm| on this input list of
length two yields two possible output values $0$ and $1$, with path conditions
$s_0 > s_1 = \mathtt{True}$ and $s_0 > s_1 = \mathtt{False}$\ifextended, respectively\fi.

\section{Testing Differential Privacy}
\label{sec:testing-differential-privacy}

To show how \fuzzidp tests differential privacy, we first review how {\em
  proofs} of differential privacy are constructed (\ref{subsec:background-aprhl}),
with ReportNoisyMax as
an example (\ref{privproof}). Then we discuss how to adapt the same basic
ideas to testing (\ref{prooftotesting}).

% We now describe the paper's core contribution: our testing
% framework. \bcp{not 100\% happy with the wording in the rest of this paragraph.}To
% begin, we briefly sketch \aprhl, a logic for formal {\em proofs} of
% differential privacy~\cite{Barthe:2016}, on which the testing framework is
% modeled. We then walk through a proof that the {\tt rnm} algorithm is
% differentially private (\ref{privproof}), highlighting several key steps,
% and finally show how to adapt the same conceptual flow to testing
% (\ref{prooftotesting}).

\subsection{Background: \aprhl}
\label{subsec:background-aprhl}

Differential privacy experts
often approach proofs in a ``pointwise'' fashion
using \Cref{defn:epsilon-differential-privacy-for-discrete-domains}: Given
two executions on similar inputs, first
demonstrate that, no matter what concrete value the first execution of
\lstinline|rnm| yields,
the second execution can also produce the same result; then show that the
multiplicative difference in the total probability of all executions that
lead to
these identical outputs is bounded by the algorithm's prescribed privacy parameter.

%\hzh{I'm not yet sure if the following paragraphs are in their right place.
%But based on Andreas's feedback, I feel we need to move up the \aprhl discussion,
%this will avoid the impression that \aprhl is an after-thought, and help putting
%the proof technique in a more solid framework, and should help readers
%understand the proof from a more PL-ish view.}
%
%\bcp{Better: ``XXX et al [citation] formalized...'', to make it clear that
%  this is prior work by others.}
\citet{Barthe:2016} formalized the pointwise proof technique in a program
logic for differential privacy called \aprhl (approximate, relational Hoare
logic). A key innovation of \aprhl is its notion of
``approximate lifting,'' which abstracts over relations between
distributions.
Approximate lifting allows us to use a \emph{deterministic} relation to
simultaneously {\em couple} %\bcp{first use: put it in quotes or in italic?}
all
samples from one distribution with samples from the
other distribution, effectively reducing probabilistic reasoning to deterministic
reasoning.  The details of \aprhl are beyond the scope of this
paper (\cite{DBLP:journals/corr/abs-1710-09951} is a readable introduction);
here we just sketch the parts on which our testing strategy is most
immediately based.

An \aprhl judgment for differential privacy has the form
$\aprhlstmt{c}{\epsilon}{\delta}{c}{\Phi}{\Psi}$, where $c$ is a
randomized program.\footnote{More generally, a relation can be established
  between syntactically different program fragments $c_1$ and $c_2$ instead
  of a single program $c$. This general form of \aprhl judgment is useful for
  intermediates steps of proofs; we ignore this refinement here for the sake of
  simplicity.}
%\bcp{for our purposes, will $c_1$ and $c_2$ ever be
%  different?}\hzh{they will be almost the same but not quite. because $c_1$
%  and $c_2$ contain references to the input database, which are variables
%  with different names but related by some neighbor relation.}\bcp{Not
%  following you yet: why different names??}\hzh{Perhaps you have the
%  annotations <1> and <2> in mind? I consider x<1> and x<2> different
%  variable names\bcp{Ah, I see.  But isn't that a bit overly
%    fiddly?}... but, perhaps more importantly, only the top-level
%  statement will be the same. Intermediate steps may differ due to
%  case-analysis, difference in control-flow paths, etc. But, I see that for
%  our presentation we only ever show the top-level stmt anyways. Would you
%  prefer having the same command showing on both sides?}\bcp{I think this
%  might be easier for people to grok.  We could add a footnote describing the
%more general form, so that experts are not confused.}.
It describes a
relation between the output distributions on related executions of
$c$.  Here, $\Phi$ is a pre-condition on the free variables of
$c$
between related executions, and $\Psi$ is a deterministic relation on the output samples of $c$ that, behind the scenes, is lifted into the corresponding
relation between the output
distributions through \aprhl's approximate lifting
machinery.
The parameters $(\epsilon, \delta)$ represent the ``cost'' of establishing the
post-condition relation $\Psi$. When $\Psi$ asserts that the related outputs of
$c$ are equal, which implies differential privacy for the program $c$, the parameters $(\epsilon, \delta)$ are exactly the privacy cost.

For example, to state that \lstinline|rnm| is $(2,0)$-differentially private
using an \aprhl judgment, we first encode the similarity relation of the inputs
in the precondition. Two similar inputs of \lstinline|rnm| must have
coordinate-wise distance bounded by $1$, which we encode with the assertion
$\bigwedge_i |xs_1[i] - xs_2[i]| < 1$. Then, to encode
differential privacy, we assert the outputs of two
executions on similar inputs are identical in the post condition:
\[\aprhlstmt
{(\mathit{out}_1 \leftarrow \mathtt{rnm}\,xs_1)} {2}{0}
{(\mathit{out}_2 \leftarrow \mathtt{rnm}\,xs_2)}
{\left(\bigwedge_i|xs_1[i] - xs_2[i]| < 1\right)}
{(\mathit{out_1}=\mathit{out_2})}.\]

Two key proof rules that formalize the pointwise proof technique
are called \textsc{Lap-Gen} and \textsc{PW-Eq}.
%\bcp{These rules are pretty dense!
%Some ideas for helping readers process them: (1) Present them one at a time
%instead of both together; (2) Find a way to make the structure of apRHL
%judgements more apparent---e.g., by making the turnstile, the tilde,
%etc. bigger, or by making them a slightly different color (dark gray?) so the
%different parts of the structure stand out; (3) put parens around the programs
%and the postcondition so it's easier to parse even without colors.}
%\bcp{These two sentences are the crux of the intuition for how the system
%  works: I think a bit more is needed.  (And I don't know what ``admits an
%  equality relation'' means.)  E.g., maybe give an intuition for where they
%  are going to get used in the proof for rnm.}
The \textsc{Lap-Gen} rule allows us to connect Laplace samples in two
executions with a deterministic relation through \aprhl's approximate lifting
theory. In
particular, this deterministic relation allows us to assume (in the
postcondition) that the samples drawn from the Laplace distribution on the
two runs are at a fixed distance $k$ apart.
\begin{mathpar}
\inferrule[Lap-Gen]
          {\epsilon = |k+e_1-e_2|/w}
          {\aprhlstmt
           {(x_1 \leftarrow \mathtt{lap}\, e_1\, w)}
           {\epsilon}{0}
           {(x_2 \leftarrow \mathtt{lap}\, e_2\, w)}
           {\Phi}
           {(x_1 + k = x_2)}
          }
\end{mathpar}
Readers encountering couplings
for the first time may worry that this proof rule looks too good to be
true, since it allows us to assume related samples are always at a
deterministic distance $k$ apart. What makes this work is that we are
not considering {\em some particular} pair of samples; rather, we are
relating {\em the entire support} from two Laplace distributions
simultaneously. Furthermore, establishing such a relation in the
post-condition does not come for free: we are allowed to choose any
$k$, but the privacy cost also depends on $k$.

This deterministic rule \ifextended allows us to build a
proof relating two runs of a randomized program while abstracting\xspace\else
abstracts\xspace\fi away direct reasoning over the related Laplace
distributions. In \Cref{privproof}, we will present an example that applies \textsc{Lap-Gen} to prove \lstinline|rnm| is $(2, 0)$-differentially private. We will refer to these $k$ values used in \textsc{Lap-Gen} rules as ``shift values.''

The \textsc{PW-Eq} rule is a formal description of the
pointwise proof technique: if we can show that, for each possible output value
$r \in \tau$, one execution returning $r$ implies that the other execution also
returns $r$,
%\bcp{There's something quite counterintuitive about this rule
%  that we should acknowledge: we're imagining a somewhat strange world in
%  which two different executions of randomized programs can be shown to
%  return the same results...}
then these pointwise facts together constitute a complete differential privacy
proof.\footnote{This rule may appear to be
  showing---counterintuitively---that different executions of a randomized
program produce identical results. Of course, any two {\em concrete}
executions of a randomized program will almost certainly produce different
results. But we are not actually reasoning here about some particular pair
of executions.
Rather, we are reasoning simultaneously about {\em all} pairs of coupled
executions.}
\begin{mathpar}
\inferrule[PW-Eq]
          {e_1, e_2 :: \tau
           \qquad \forall
           r \in \tau, \;\;
           \aprhlstmt
            {c_1}
            {\epsilon}{0}
            {c_2}
            {\Phi}
            {(e_1 = r \rightarrow e_2 = r)}}
          {\aprhlstmt
           {c_1}
           {\epsilon}{0}
           {c_2}
           {\Phi}
           {(e_1 = e_2)}}
\end{mathpar}%
The variables $e_1$ and $e_2$ here represent the
output values of $c_1$ and $c_2$ respectively.
The power of \textsc{PW-Eq} is that, for each
$r \in \tau$, we are allowed to choose a different shift value for
any application of the \textsc{Lap-Gen} proof rule that may appear in the
subderivation for this $r$. To prove \lstinline|rnm| is
$(2, 0)$-differentially private, we will apply exactly this strategy: given
any output $r$ from one execution of \lstinline|rnm|, we choose some
sequence of shift values so that the noisy max also occurs at $r$ in the
second execution, forcing the second execution to also return $r$.

With these two proof rules, the process of proving differential privacy for
ReportNoisyMax reduces to a simpler goal: for each possible output of
ReportNoisyMax, find a sequence of shift values to relate the Laplace
samples between two executions so that their output values are
identical.

\Cref{prooftotesting} will show how we use the ideas behind the
\textsc{Lap-Gen} and \textsc{PW-Eq} rules in designing our testing strategy
for differential privacy.  But first, a concrete
example.

%\iffalse
%However, it is not obvious how proofs of this style
%imply \Cref{defn:epsilon-differential-privacy}, as the original definition
%requires the probability bound to hold for all possible set of outcomes, while
%this new proof focuses on one possible output at a time. To bridge this gap,
%researchers developed metatheorems certifying the validity of proof techniques
%in this style \cite{DBLP:journals/corr/abs-1710-09951,
%DBLP:journals/corr/abs-1904-12773,
%DBLP:journals/corr/abs-1710-09010}. Furthermore, these metatheorems are
%developed within frameworks that also abstract over the accounting of privacy
%parameters with a notion of ``cost''. In these frameworks, establishing a
%relation between samples from probability distributions incurs a non-negative
%cost, and the total cost of a proof of differential privacy is tallied up by a
%composition theorem within the framework, and the value of this total cost
%corresponds to the privacy parameter of the algorithm.
%
%For this work, we use the work by \citet{DBLP:journals/corr/abs-1710-09951} as
%the theoretical foundation for reasoning with differential privacy since its
%theory and its generalization developed
%by \citet{DBLP:journals/corr/abs-1710-09010} are compatible with how our testing
%framework validates differential privacy for \fuzzidp programs.
%\fi

\subsection{An Example Privacy Proof}
\label{privproof}

%\bcp{Maybe it would be good to point out where the two key rules are used in
%this proof...}
%
Let's consider a (paraphrase of a) proof by
\citet{Dwork:2014:AFD:2693052.2693053} of $(2, 0)$-differential privacy for
\lstinline|rnm| when the coordinate-wise distance  (\Cref{defn:coordinate-wise-distance-relation}) between the inputs is
bounded by $1$. We present
this proof by sketching applications of \textsc{Lap-Gen} and \textsc{PW-Eq}, and emphasizing steps in the proof that will
become key
ingredients in our testing algorithm by putting boxes around equations.
This proof can be carried out formally in \aprhl, but we
present it informally to avoid getting bogged down in details.

\begin{thm}
ReportNoisyMax is $(2, 0)$-differentially private.
\end{thm}

\begin{proof}
%For ReportNoisyMax, two input lists are similar if
%\begin{enumerate}
%\item the two lists $qs_1$ and $qs_2$ have the same length;
%\item the difference $|qs_1[i] - qs_2[i]|$ is less than $1$ for all index $i$.
%\end{enumerate}
The implementation of ReportNoisyMax creates a list of noised values based on
its input. Let $r$ be any possible output of
running \lstinline|rnm|. Let $\mathtt{argmax}$ be the function that
returns the index of the largest value in a list. When \lstinline|rnm| runs on
an list of input values, it first adds Laplace noise to each of these values,
then iterates over the noised values, keeping track of the index of the maximum
value seen so far, and finally returns that index. Write $qs_1'$ and $qs_2'$
for the two intermediate lists of noised values from $qs_1$ and $qs_2$.
Then, if $r$ is the result of
running \lstinline|rnm| on $qs_1$, it is easy to see that $r =
\mathtt{argmax}\, qs_1'$.

We next apply the pointwise proof technique described
in \Cref{subsec:background-aprhl}. Since we assumed one execution
of \lstinline|rnm| returned index $r$, we need to show some control flow through
the run of \lstinline|rnm| on $qs_2$ yields the same $r$. We demonstrate such a
control flow by carefully choosing the shift values introduced in
the \textsc{Lap-Gen} proof rule. Concretely, we need to ensure $r
= \mathtt{argmax}\,qs_2'$, under some choice of the shift values that connects
each $qs_1'[i]$ and $qs_2'[i]$.

We proceed by applying \textsc{Lap-Gen} to each pair of noised values with the following choice
of shift values, coupling each $qs_1'[i]$ with $qs_2'[i]$ such that $qs_2'[i] = qs_1'[i] + \mathit{shift}_i$:
%\bcp{First, this is not defining the shift values; it is
%  defining the second noised list.  Second, $qs_2'$ has already been defined
%above; it doesn't make sense to (appear to) redefine it here.}
\boxedeq{eq:rnm-shift}{
\mathit{shift}_i \ = \
\begin{cases}
  1 &\text{if } i == r\\
  qs_2[i] - qs_1[i] &\text{otherwise}
\end{cases}}%
%\bcp{no need to capitalize: sentences containing math displays should
%  (generally) be punctuated exactly as if the displayed material were inline.}
This choice implies that the maximum value in the second noised array also
occurs at index $r$: we know that
\ifextended$$\else$\fi
-1 < qs_2[i] - qs_1[i] < 1
\ifextended$$\else$\xspace\fi
for all $i$, since $qs_1$ and $qs_2$ are
similar inputs. In particular,
$qs_2[i] - qs_1[i] < 1$.
Adding $qs_1'[i]$ to both sides of this inequality yields
$$qs_1'[i] + (qs_2[i] - qs_1[i]) < qs_1'[i] + 1.$$
Since $qs_2'[i] = qs_1'[i] + (qs_2[i] - qs_1[i])$ if $i \neq r$ and
$qs_2'[r] = qs_1'[r] + 1$, {it follows that $r
= \mathtt{argmax}\, qs_2'$}.

We have shown so far that, for any output index $r$ from the first
execution, it is
possible for \lstinline|rnm| to produce the same result $r$ on the similar
input. Next we need to calculate the $\epsilon$ privacy cost incurred by using
\textsc{Lap-Gen} for connecting $qs_1'$ with $qs_2'$, and prove $\epsilon$ is
at most $2$.

Consider the cost of an application of \textsc{Lap-Gen}
%\bcp{how many
%  uses of Lap-Gen are there in this proof?  I had assumed one
%  (because there's just one call to Laplace in the loop body) or maybe two
%  (because we need a case split on whether $i = r$), but the
%  wording sounds like there are more than one?}
%  \hzh{there are $n$ Lap-Gen applications, where $n$ is the length of the input list. There are some laplace calls hidden by the mapM operator, which calls laplace once for each value in that list.}
  between a given pair of
$i^{\rm th}$ sample values $qs_1'[i]$ and $qs_2'[i]$ in the two runs; call this $\mathit{cost}_i$. Under \aprhl, the total
privacy cost is bounded by the sum of the $\mathit{cost}_i$s, so we need to
show $\sum_i \mathit{cost}_i < 2$.

To bound the sum, let us first consider the cost values for indices $i \neq
r$.  We know $qs_2'[i]$ is a sample from the
distribution $\lap{qs_2[i]}{1.0}$, and similarly $qs_1'[i]$ is a sample from
$\lap{qs_1[i]}{1.0}$. Using \textsc{Lap-Gen}, we
can conclude
\begin{align*}
\mathit{cost}_i &= |\left((qs'_1[i] + qs_2[i] - qs_1[i]) - qs_2[i]\right) -
(qs'_1[i] - qs_1[i]) | \\
&= 0.
\end{align*}
On the other hand, if $i = r$, let $s = qs_1'[r] - qs_1[r]$.  We again
use \textsc{Lap-Gen} to conclude:
\boxedeq{eq:rnm-cost}{
\mathit{cost}_r &= |(qs_1'[r] - qs_1[r]) - (qs_2'[r] - qs_2[r])| \nonumber\\
                &= |s - (qs_1'[r] + 1 - qs_2[r])| \hspace{4.5em}\nonumber\text{by \cref{eq:rnm-shift} and assumption of $s$}\\
                &= |s - (qs_1[r] + s + 1 - qs_2[r])| \hspace{3em}\nonumber\text{by assumption of $s$}\\
                &= |(qs_1[r] - qs_2[r]) + 1| \nonumber\\
                &< 1 + 1 = 2 \hspace{10.7em}\text{by the triangle inequality}
}%
Since cost values are $0$ for indices $i \neq r$ and less than $2$ when
$i = r$, the total cost is less than $2$. That is, \lstinline|rnm| is $(2,
0)$-differentially private.
\end{proof}

\subsection{From Proving to Testing}
\label{prooftotesting}

The key steps in the above proof are:
%\bcp{Either punctuate with periods or
%  use lower-case letters at the beginning of each.}
\begin{enumerate}
\item  Assume an unknown but fixed output $r$ from one execution.
\item  Select a sequence of shift values (\cref{eq:rnm-shift}) to connect the samples from one execution with samples from the other execution.
\item  Show that the second execution, whose Laplace samples are fixed through the shift values, leads to the same output~$r$.
\item  Compute the total privacy cost of this pair of executions as the
sum of the individual cost
values induced by the chosen sequence of shift values and show this total
is less than the prescribed $\epsilon = 2$ (\cref{eq:rnm-cost}).
\end{enumerate}

In order to convert this to a testing procedure, we need to check that, for
a run
of \lstinline|rnm| on $qs_1$, there exists a dual execution of \lstinline|rnm|
on $qs_2$ leading to the same output, such that the Laplace samples in these two
executions can be connected through a sequence of shift values, while keeping
the total privacy cost induced by the shift values under $2.0$. However, we need
to be careful about how many distinct sequences of shift values are allowed for
constructing dual executions. The \textsc{PW-Eq} rule allows us to
choose one sequence of shift values for each unique output $r$
that \lstinline|rnm| may return. Thus, for testing, we must group the runs
of \lstinline|rnm| on $qs_1$ by their output, and then try to find
a single sequence of shift values for each group, such that this single
sequence of shift values leads to corresponding executions of
%\bcp{on each of
%the inputs in a given group}\hzh{there're only two neighboring inputs $qs_1$ and $qs_2$.}
\lstinline|rnm| on $qs_2$
with the same output.

With this in mind, the first step in testing is to run \lstinline|rnm| on
$qs_1$ a large number of times
and group the runs by their final output values.
%
%In order to obtain a hypothesized\bcp{This is funny wording: if we observe a
%given $r$ by actually running the program, then it is not hypothesized, is
%it?  More generally, I feel like the transition here is too fast.  Needs
%something along the lines of ``in order to convert this to a testing
%algorithm, here's what we need to think about...''} output
%value $r$ for testing, we can
%repeatedly evaluate \lstinline|rnm| for $N$ times on some input value $qs_1$ and
%observe its set of output values. To check the differential privacy property
%of \lstinline|rnm|, we need to perform analogs of steps 2, 3, and 4 for each
%unique $r$ we observe.  \bcp{Is there a quick intuition for why doing this
%  constitutes a persuasive test?}
%

Next we need to find the shift values.
In steps 2 and 3 of the proof above, we used expert insight to select shift values
that allowed us to show that the dual execution must result in the same output
$r$. When testing, we begin with hypothetical (symbolic) shift values and
hope to work out their concrete values later: we create one symbolic shift value for each $qs'_1[i]$ per group,
and pair $qs'_1[i]$ with a $qs'_2[i]$ through the following
equation
%\bcp{It's still not clear what is the force of writing down this
%  equation.  Is it a
%  constraint that the tester will have to solve later?  A definition of
%  $qs'_2[i]$?  Something else?}\hzh{Yes, it is part of the constraints that Z3 will solve later.}
\boxedeq{eq:generic-shift}{qs'_2[i] = qs'_1[i]
+ \mathit{shift}_i.}
This equation is exactly how \textsc{Lap-Gen} allows us to
connect two Laplace samples. (Throughout this discussion, we highlight the important
symbolic formulas created in the testing steps by boxing them.)
%\bcp{I stuck that
%sentence here because I couldn't find a better place for it; but I wish there
%were a higher-level discussion it could be part of.}
This also induces a symbolic cost using \textsc{Lap-Gen},
where
%\bcp{be careful of the
%  spacing after this display: if you look carefully, you'll see that there's
%an extra space before ``We''...} \bcp{It's still there---it's coming from
%the newline after the closing curly brace.  If you put a \% right after the
%curly brace it will look better.}
%
\boxedeq{eq:symbolic-cost}{
\mathit{cost}_i &= \frac{|(qs'_2[i] - qs_2[i]) - (qs'_1[i] - qs_1[i])|}{w}\nonumber \\
                &= \frac{|qs_1[i] + \mathit{shift}_i - qs_2[i]|}{w}.
}%
Here $qs_1[i]$ is the concrete center supplied to the $i$th call to
the sampling instruction in the first execution, $\mathit{shift}_i$ is a fresh
symbolic variable, $qs_2[i]$ is the
%\bcp{symbolic(?)}
concrete center used in the $i$th sampling
instruction in the second execution, and $w$ is the (known constant)
parameter controlling the
width of the Laplace distributions used in \lstinline|rnm|. Thus, both
$qs'_2[i]$ and $\mathit{cost}_i$ can be represented as symbolic
expressions coupled with concrete execution traces using the formulas above if we know what $qs_1[i]$ and $qs_1'[i]$ are.
%\bcp{Don't
%  follow the reasoning here: Why could they not be adequately represented
%  otherwise?}\bcp{Still don't follow: why can't they be represented as
%  formulas even if $qs_1[i]$ and $qs_1'[i]$ are unknown?}
%\hzh{how about now? The point I'm trying to get across is that the symbolic representation for $qs_2'$ is only meaningful/useful for testing if they are tied to observed sampled traces.}

To capture $qs_1$ and gather many independent samples of $qs_1'$ for testing, we repeatedly run \lstinline|rnm| on $qs_1$ with a special
%hzh: this interpreter is not symbolic, it only instruments
interpreter that instruments Laplace sample instructions,
recording the
center, width, and returned sample value for each.
We group each unique output $r$ together with all sequences of Laplace sample and
parameters that lead to the output $r$ into a \emph{bucket}.
%\bcp{the
%  beginning of the paragraph was talking about a single run; the end is
%  talking about many runs. The implicit transition from one to the other is
%  confusing.}

%\bcp{This sentence is totally confusing. Also, we just introduced buckets,
%  but this paragraph seems to introduce them again.  These two paragraphs
%  need to be reorganized.}
%\bcp{This is better, but still confusing: for one thing, there are two
%  ``However''s in a row---never a good sign.  Also, the order of quantifiers
%is puzzling: where are we talking about the runs from a single bucket,
%were are we quantifying over all buckets, etc.?}
For each bucket $\mathit{bkt}$, we need to show there exists a coupling that connects all sampled traces in $\mathit{bkt}$ to the Laplace samples from running on $qs_2$, such that this coupling leads the dual execution on $qs_2$ to the same result $r$ from $\mathit{bkt}$. To find such a coupling, we symbolically execute \lstinline|rnm| with symbolic samples
$qs'_2$, observing all possible outputs of \lstinline|rnm| along all of the
control flow paths. Among these paths, some will return $r$. We gather all of
the path conditions from control flow paths that lead to $r$. These
path conditions are then used to constrain the shift values
from \cref{eq:generic-shift}, so that the coupled dual execution only takes the paths that yield the same output $r$.
%
%For ReportNoisyMax, these path conditions are exactly the necessary and
%sufficient conditions for ensuring its differential privacy property. We can
%consider these path conditions a natural specification over the generic shift
%relation.\bcp{Don't understand that sentence.}
%
We use {$\Phi_r$} to denote the symbolic formula that encodes the path
conditions, and {$\Omega_r$} to denote the symbolic formula that encodes the
shift equations (\cref{eq:generic-shift}). These two symbolic formulas form the
testing analog of step 2 and 3 from the proof.

Finally, to bound the total privacy cost (as in step 4 of the proof),
we create another symbolic formula using the symbolic expressions for each
$\mathit{cost}_i$:
\boxedeq{eq:total-cost}{
\sum_i \mathit{cost}_i \leq 2
}

Together, the boxed symbolic formulas constitute a query that can be dispatched
to an off-the-shelf SMT solver---we use
Z3 \cite{DeMoura:2008:ZES:1792734.1792766}. If the solver returns a satisfying model
for these constraints, then we know the distributions produced by
running \lstinline|rnm| on this particular pair of $qs_1$ and $qs_2$ likely
satisfy 2-differential privacy. Of course, this does not guarantee
differential
privacy, because we used \emph{sampled traces} of Laplace calls to describe
properties of a potential proof for differential privacy, instead of universally
quantifying over all possible samples in one execution,  as in the proof.

At the core of both the proof and the SMT formula is the relation between
$qs'_1$ and $qs'_2$. The proof demonstrates that there exist
shift
values such that the related Laplace samples satisfy the
privacy cost bounds. The
testing process also checks for the existence of such shift
values; however, it has a chance of admitting programs that
are not $(\epsilon, 0)$-differentially private, because testing does
not produce a complete proof of differential privacy---it only checks whether dual executions
satisfying differential privacy exist on some set of sampled traces.
There are two important test parameters---number of pairs of randomly generated similar inputs, and number of sampled traces collected on a given pair of similar input---that can be independently tuned to make tests more difficult to pass for faulty programs. Intuitively, as we
increase both test parameters,
%\bcp{What is the difference between ``increase the number of sampled
%  traces'' and ``increase the number of tests''?}\hzh{Each test generates 1
%  pair of neighboring inputs $qs_1$ and $qs_2$. Within that test, we run the
%  program on $qs_1$, say 1000 times, collecting 1000 sampled traces, and
%  group them into buckets. So, if let number of tests be $d$, and number of
%  sampled traces within each test be $n$, then total number of sampled
%  traces across all $d$ tests will be $d\cdot n$. The two parameters $d$ and
%  $n$ can be independently tuned.}\bcp{Ah: we should spell this out explicitly!}
it should be less and less
likely that our testing process accepts a faulty program. We call
programs that are faulty but slip past \fuzzidp's testing framework {\em false
negatives}.

A non-$(\epsilon, 0)$-differentially private program may be faulty for two
reasons: 1) it may have a non-zero $\delta$ failure probability, and 2) there
may exist similar inputs $(x_1, x_2)$ for which the program produces
distributions that do not satisfy \cref{defn:epsilon-differential-privacy-for-discrete-domains} with the given $\epsilon$. Note
that the definitions of differential privacy
(\ref{defn:epsilon-delta-differential-privacy}
and \ref{defn:epsilon-differential-privacy-for-discrete-domains}) require the
relation on output distributions in \cref{defn:epsilon-differential-privacy-for-discrete-domains} to hold on all similar
inputs. Although \fuzzidp's testing framework can check that \cref{eq:1}
holds on a
pair of fixed similar inputs by checking the existence of shift values and dual
executions, it can never exhaustively check that \cref{eq:1} holds on a
potentially
infinite set of similar pairs of inputs.

Instead, \fuzzidp gives probabilistic guarantees using the framework of random
differential privacy \cite{Hall2013}. In the next section, we review random
differential privacy and use it to state our main
theoretical testing guarantee (\Cref{thm:sample-complexity}).

\subsection{Implementation}
\label{subsec:implementation}
We now describe the testing process more concretely, showing type
signatures of Haskell functions that implement key testing steps.
%  and
% discussing some tradeoffs involved in combining instrumented execution and
% symbolic execution
% BCP: The discussion is pretty short -- better not to set readers up to be
%   disappointed.

\fuzzidp's testing framework takes as inputs a program under test,
$\mathtt{prog} :: \sigma \rightarrow \mathtt{Expr}\, (\distr\, \tau)$, a
generator, $\mathtt{gen} :: \mathtt{Gen}\,(\sigma, \sigma)$,\footnote{A value of type $\mathtt{Gen}\,\tau$ is a function that takes a seed and produces a pseudo-random value of type $\tau$.}
%\bcp{Has ``Gen''
%  been explained??}
  that produces pairs
of similar inputs for \lstinline|prog|, and a privacy parameter $\epsilon$. It
then checks, for a large number of $(x_1, x_2)$ pairs produced
by \lstinline|gen|, that the distributions produced by running $\mathtt{prog}\,
x_1$ and $\mathtt{prog}\, x_2$ satisfy \cref{eq:1}.  If it ever finds one that
does not, it rejects \lstinline|prog|; otherwise it validates \lstinline|prog|
as likely to be $(\epsilon, 0)$-differential private. (We discuss what
``likely'' means more formally in
\Cref{sec:random-differential-privacy}. We also show how much time \fuzzidp
takes to reject buggy benchmark programs in \Cref{fig:evaluation} in the appendices.)
%\bcp{We could also say that we
%  quantify it in the evaluation section, but do we? I.e., do we say how many
%tests / how much testing time it takes to reject buggy versions of things?
%That would be interesting information!}\hzh{Hmm isn't that the main theorem?
%given an ill-behaved program (quantified using RDP), the probability of
%rejecting it is at least $1 - \exp{(\cdots)}$(the exponentially small
%probability in the main theorem) as long as the test size and count satisfy
%the bounds in the main theorem.}\bcp{But don't we say that the main theorem
%is not a practical result, given our scaling difficulties?  And either way,
%isn't it interesting to know {\em how fast} we can identify buggy mutants in
%practice?}\hzh{Hmm yes. We do have a table that shows how many tests/how
%much time it takes \fuzzidp to reject buggy programs in practice. I will
%leave a pointer here.}\bcp{Since that table is in an appendix, we could briefly
%summarize here in the text.}
The more tests a program
passes, the more likely the program really is $(\epsilon, 0)$-differentially
private.
%\bcp{Maybe this is more detailed than we need here?}
Our experiments on benchmark algorithms
%\bcp{pointer to appendix with details}
show that \fuzzidp rejects
many faulty programs within 10 seconds, but it may take significantly longer
testing time to reject algorithms that only demonstrate privacy violations
on larger inputs\ifextended: one of the faulty variants of SparseVector takes on
average 218 seconds until rejection, and a faulty variant of SparseVectorGap
with a similar bug takes on average 1899 seconds until rejection\fi\xspace(see \Cref{ap:detailed-evaluation} for detailed benchmark study).
%\bcp{Should there be an $\alpha$ in that claim??}
%\hzh{I don't think so... the $\alpha$ is hidden by ``likely'', and the exact $\alpha$ depends on how many tests we run}
%\bcp{Maybe add a note about this?}
%  for
% $(\epsilon, 0)$-differential privacy, and reject $\mathtt{prog}$ otherwise

To verify that the distributions produced by a particular pair of similar
inputs are
related, we need to construct a coupling
%\bcp{better without quotes,
%  actually---readers know what it is by this point}
  between the Laplace distribution
samples used by $\mathtt{prog}\, x_1$ and $\mathtt{prog}\, x_2$ using the
methods described in \Cref{sec:testing-differential-privacy}. We first acquire
concrete sample values from some large number (call it $N$) of runs of
$\mathtt{prog}\, x_1$.\footnote{For
invalidating incorrect algorithms, we can start with small $N$ (such as $50$)
and keep increasing $N$ by a factor of $10$ until the bug is discovered. For
validating correct algorithms, $N$ should be chosen according to the lower bound
on $m$ in \Cref{thm:sample-complexity}. However, currently the computational
cost of running tests with large $N$ makes validation prohibitively slow. We
discuss this issue in \Cref{sec:limitation}.}
We denote the output from the $i^{\mbox{\tiny th}}$
run by $r_i :: \tau$ and the sample trace from each of these runs by
$\mathit{tr}_i :: \mathtt{Trace}$, where
%\bcp{Usually I think of a trace as a list; here we are using the
%word for the things in the list.  Might be confusing.}
%\bcp{Might be simpler for some readers if we just make it a
%  triple?}
\begin{lstlisting}
type SampleInfo = (Double, Double, Double)
type Trace      = [SampleInfo]
\end{lstlisting}
and where the following projections extract sample, center, and
width values from a \lstinline|SampleInfo|:
\begin{lstlisting}
sample, center, width :: SampleInfo -> Double
\end{lstlisting}
With the collected outputs $r_1$, $r_2$, \dots, $r_N$ and traces
$\mathit{tr}_1$, $\mathit{tr}_2$, \dots, $\mathit{tr}_N$, we perform a
``bucketing''
process so that all traces that lead to the same output value $r$ are grouped
together.
\begin{lstlisting}
type Buckets &$\tau$& = Map &$\tau$& [Trace]
bucket :: [(&$\tau$&, Trace)] -> Buckets &$\tau$&
\end{lstlisting}
A value of this map type represents a collection of buckets; a particular
key-value pair (of an output value with its associated list of traces) is a
single bucket.

Next, we perform symbolic execution on $\mathtt{prog}\, x_2$. To do this, we
first perform a simple program transformation:
$\mathtt{streamline} :: \mathtt{Expr}\, (\distr\, \tau) \rightarrow
[\mathtt{Expr}\, (\distr\, \tau)]$. This transformation repeatedly replaces a
program containing
\lstinline|If| commands with two new programs in which the \lstinline|If|
command is replaced by sequencing \lstinline|Assert cond| with the commands
in the true branch and sequencing \lstinline|Assert (not cond)| with
commands in the false branch.

\begin{lstlisting}
data Expr a where
  ...
  Assert :: Expr Bool -> Expr (&$\distr$& ())
  Sequence :: Expr (&$\distr$& ()) -> Expr (&$\distr$& a) -> Expr (&$\distr$& a)
\end{lstlisting}
This simplistic approach produces $2^n$ straight-line programs in the worst
case, where $n$ is the number of \lstinline|If| statements; we will discuss a
type-driven optimization adapted from \cite{Torlak:2014:LSV:2666356.2594340} for
speeding up symbolic execution in \Cref{sec:optimizations}.

Note that \lstinline|streamline| would actually diverge
%\bcp{``would diverge'' means
%  ``does diverge''...?
%Or ``would diverge except that...''?  Or ``would diverge if Haskell were not
%lazy''?  Or?}
on infinite syntax trees if Haskell were not lazy. Our
 symbolic interpreter uses information gathered by instrumented
implementation to cut off infinite symbolic execution as soon as
possible.
%\bcp{Why is this not ``handling them in general''?}.
We
use this early cutoff trick
%\bcp{How ad hoc is it really?  Either this is selling
%  ourselves short (if it would work for a reasonably large class of
%  algorithms) or it is a significant limitation (i.e., we've designed an
%  analysis framework that works exactly for previous algorithms + PrivTree
%  but nothing new)...}
  in the evaluation of PrivTree; the trick and some directions for generalizing
  it are discussed in more detail in \Cref{sec:limitation}.
%\bcp{Wouldn't it be simpler to make ``assert'' take a
%second argument (representing the rest of the program)?}
%\hzh{Hmm true... I think I was trying to make everything fit into monad
%combinators as much as possible when I designed it. This results in code that
%looks pretty close to imperative programs when prettyprinted. But I agree having
%Assert take another argument does seem simpler.}
%\bcp{I tried to rewrite that sentence, but I'm
%  still confused by what it should say.  The ``If'' constructor doesn't just
%  take one argument: it takes three.  What happens to the ``then'' and ``else''
%  arguments after the transformation?}
%\bcp{I also have a technical question: If
%there is an If inside a loop, doesn't the loop have to be completely unrolled
%to produce this list of straight-line programs??}
%\hzh{You are right, and that's why symbolic execution cannot handle infinite
%  syntax tree. Laziness really only helps regular interpretation.}
%\bcp{Deserves a comment here about how we deal with that.}

The programs resulting from this transformation are free of conditional
branches, instead explicitly encoding path conditions using
\lstinline|Assert| nodes. We next take these transformed programs and run
symbolic execution guided by the trace buckets from the
instrumented executions above.

Consider a particular set of executions that lead to the output value $r$, and
let the associated trace bucket contain the traces
$\mathit{tr}_{i_1}, \mathit{tr}_{i_2}, \dots, \mathit{tr}_{i_k}$, where $i_1,
i_2, \dots, i_k$ are the indices of instrumented runs that produced $r$. We then
search for paths that produce the same output $r$ and
build \cref{eq:generic-shift} between the concrete sampled traces and the
symbolic Laplace samples. For each trace $\mathit{tr}_{i_k}$, we pair it with
symbolic Laplace samples as follows: on the $j$th call to the Laplace sampling
instruction during symbolic execution, we create a fresh symbolic value
$\mathit{lap}_{i_k}[j] = \mathtt{sample}\, (\mathit{tr}_{i_k}[j])
+ \mathit{shift}_j$. Let $\psi_{k} = \bigwedge_j \mathit{lap}_{i_k}[j]
= \mathtt{sample}\, (\mathit{tr}_{i_k}[j]) + \mathit{shift}_j$, and $\Psi_r
= \bigwedge_k \psi_k$, and let $\Phi_r$ encode the disjunction of the path
conditions for all control flow paths that lead to the output $r$.

The final formula
$\Psi_r\, \land\, \Phi_r\, \land\, \left(\sum_{n} \mathit{cost}_n \leq \epsilon\right)$
asserts that $\mathtt{prog}$ produces the same
output $r$ within the prescribed privacy cost $\epsilon$. We perform the same
process for each unique output $r$ observed from instrumented executions. If
these formulas are all satisfiable, we consider this test a passing test case,
and we do not reject the claim of $(\epsilon, 0)$-differential privacy.
%We then
%repeat the process using another randomly generated pair of data until enough
%tests pass.
%\hzh{I understand this is a
%probably not a strong way of choosing these numbers... But I think maybe
%reporting them as a fact could make them seem less arbitrary.}\bcp{Yes, this is better.}}.
On the other hand, if, for some output $r$, Z3 tells us that the formula is
not satisfiable, then
the program under test does not have a point-wise proof of
$\epsilon$-differential privacy using the proof template discussed
in \Cref{sec:testing-differential-privacy}. This is not a {\em dis}proof of
differential privacy, since the proof template we are using is not complete
(there are algorithms whose privacy proofs do not follow this pattern), but it
is at least a signal that should prompt us to look at the program
skeptically.

The code below sketches the testing process on an input procedure to
test \lstinline|prog|, a pair of neighboring inputs \lstinline|x1|
and \lstinline|x2|, an expected privacy parameter \lstinline|eps|, and the
number of sampled traces to draw for testing
\lstinline|ntraces|.
%\bcp{Readers may wonder whether ``assert'' here is the
%  same as the Assert on the previous page.}
%
\begin{lstlisting}
expectDP prog x1 x2 eps ntraces = do
  buckets <- instrumentedExec ntraces (prog x1)
  let straightlineProgs = streamline (prog x2)
  constraints   <- symbolicExec buckets straightlineProgs
  solverResults <- runSolver eps constraints
  expect (all isOk solverResults)
\end{lstlisting}
%
\iffalse
\iflater
The testing process is illustrated in
\Cref{fig:testing-process}. \bcp{Besides the figure using too much space, I don't think
it's very useful to refer to it only so late in the section.  If we keep it,
it should be used as a point of reference in the whole discussion.}

\begin{figure}[t]
\centering
\includegraphics[width=0.45\textwidth]{flowchart.png}
\caption{Testing process \bcp{Not sure we can afford the space for this
    picture.}}
\label{fig:testing-process}
\end{figure}
\fi
\fi

A distinctive element of our design is the combination of instrumented
execution and
symbolic execution. An alternative would be to run symbolic execution on
{\em both} inputs
using relational symbolic execution
\cite{DBLP:journals/corr/abs-1711-08349}, then universally quantify over the
Laplace samples in one execution, as
demonstrated in the analysis of \lstinline|rnm|
from \Cref{sec:testing-differential-privacy}. Using relational symbolic
execution would bring more confidence to the differential privacy property of the
program under test, since the satisfying model from the SMT solver will serve as
a formal \emph{proof} of \cref{eq:1} for the pair of output distributions. However,
this approach produces more complex symbolic formulas that may significantly
slow down Z3.
To strike a balance between execution time and confidence gained
from a passing test, we choose to combine instrumented execution and symbolic
execution.

\section{Asymptotic Privacy Guarantees}
\label{sec:random-differential-privacy}
%\bcp{Now I'm wondering if ``Correctness'' isn't selling this section short:
%  it is not just about correctness, but about (at least theoretical) {\em
%    guarantees}.  So maybe something like ``Asymptotic Privacy Guarantees''?}
% \bcp{How about calling this section ``correctness''?  (But it's a little
%   weird to talk about correctness before explaining the Design...)}
% \hzh{How about testing guarantees?}
% OK
% \bcp{Actually, I think I like ``Correctness'' better now, as a name for this
% section.  And then maybe the next section can be called
% Implementation?}\hzh{Correctness sounds good. But we do have a section 9
% that's also Implementation that goes into a bit more details about other
% (non-testing) parts of the implementation... should we rename that to
% something else? ``Implementation Details''?} \bcp{Not sure how much that
% helps.  I think the real problem is that we do not have a clear reason for
% putting things in section 5 vs. 7 vs. 9.  One way to solve this would be to
% merge 7 into 5 as a final subsection (maybe called Implementation).  The
% first paragraph of 9 belongs earlier, where the language is introduced, in
% any case.  And the rest of 9 could be called Optimizations, maybe?}
%\bcp{Needs a sentence or two of introduction: right now it feels like we're
%  dropped in the middle. In particular, why is random DP the right property to
%  ask for here?}
In this section, we introduce random differential privacy (RDP), and use RDP to quantitatively define ``well-behaved'' and ``ill-behaved'' programs. We present \Cref{thm:sample-complexity}, which gives an upper bound on the probability of \fuzzidp falsely accepting ill-behaved programs. However, \fuzzidp has a scaling bottleneck that prevents us from applying \Cref{thm:sample-complexity} to produce meaningful guarantees of random differential privacy on correct programs. We discuss details of the scaling bottleneck and its implication on statistical validation of correct programs in \Cref{sec:limitation}.

To introduce random differential privacy, let us first consider an alternative
view of the privacy parameters $\epsilon$ and $\delta$.
\begin{defn}
Let $\mu_1, \mu_2 :: \distr\,\tau$ be two distributions with identical
support. Define a function $f :: \tau \mapsto \mathbb{R}$. Let $f(v)
= \ln \frac{\mu_1(v)}{\mu_2(v)}$.  The {\em privacy loss random variable} is a
distribution $\mathtt{pv}(\mu_1, \mu_2) \in \distr\,\mathbb{R}$ defined
as: %$f(x)$ where $x\sim \mu_1$.
$\mathtt{pv}(\mu_1, \mu_2)(x) = \sum_{v \in \mathtt{supp}(\mu_1)\, \mathit{s.t.}\, f(v)\, =\, x} \mu_1(v)$.
\end{defn}
Informally, this distribution can also be expressed as $f(v)$, where the random variable $v \sim \mu_1$.
%\bcp{Redundant?}The privacy loss random variable represents the implicit distribution formed by
%taking the pointwise log ratio of probability difference of the two
%distributions $\mu_1, \mu_2 :: \distr\,\tau$.
\begin{defn}[\cite{Kasiviswanathan2014}]
\label{defn:epsilon-delta-pointwise-indistinguishable}
Two distributions $\mu_1$ and $\mu_2$ are {\em $(\epsilon, \delta)$-pointwise
indistinguishable}
%\bcp{make sure to do a careful spellcheck at the end}
if the probability mass of $\mathtt{pv}(\mu_1, \mu_2)$ in
the interval $[-\epsilon,\epsilon]$ is at least $1-\delta$.
\end{defn}
Note that $(\epsilon, \delta)$-pointwise indistinguishability implies
\cref{eq:1} in the definition of differential privacy
(Definition \ref{defn:epsilon-delta-differential-privacy}). Furthermore, the proof
template we introduced in \Cref{sec:testing-differential-privacy} constructs
proofs of $(\epsilon, 0)$-pointwise indistinguishability.
\begin{defn}[\cite{Hall2013}]
\label{defn:random-differential-privacy}
Assume a fixed distribution over similar inputs $\mathcal{I} :: \distr\,
(\sigma \times \sigma)$. A randomized program $f :: \sigma \mapsto \distr\,\tau$
is {\em $(\epsilon, \delta, \alpha)$-random differentially private} if, with
probability at least $1-\alpha$, sampling similar inputs $(x_1, x_2)$
%\bcp{not certain what ``with
%probability at least $1-\alpha$ over the draws of similar inputs $(x_1,
%x_2)$'' means}
from $\mathcal{I}$ leads to $(\epsilon, \delta)$-pointwise indistinguishable
distributions $f(x_1)$ and $f(x_2)$.
%\bcp{This is getting
%  pretty thick for me.  Maybe also for other PLDI readers?}
%\hzh{maybe a visual graph would help?}
\end{defn}

We can give some intuition for \Cref{defn:random-differential-privacy} by
considering visual graphs of the privacy loss random variable under given
similar inputs. First, assume a distribution $\mathcal{I}$ of similar inputs. Let
us sample nine pairs of similar inputs from $\mathcal{I}$ and draw their privacy
loss random variables as a graph centered at $0$ on the horizontal axis. We
shade each
graph with blue if its area is at least $1-\delta$ in the interval
$[-\epsilon, \epsilon]$, and red otherwise.

%\begin{center}
%\includegraphics[width=0.45\textwidth]{privloss.png}
%\end{center}
\begingroup
\setlength{\columnsep}{2pt}%
\setlength{\intextsep}{0pt}%
\begin{wrapfigure}{r}{0.4\linewidth}
\centering
\includegraphics[width=0.9\linewidth]{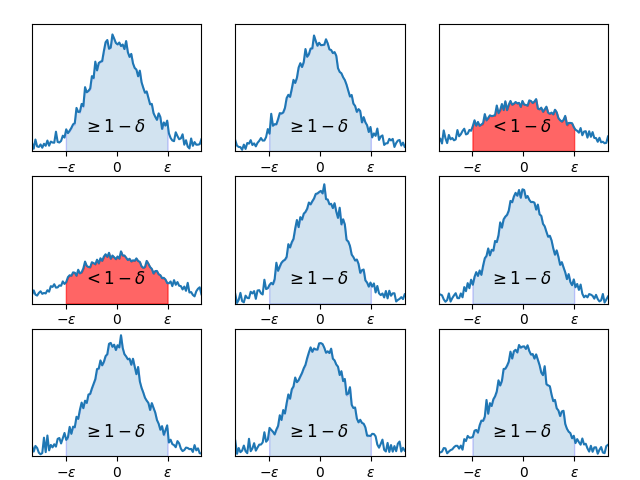}
\end{wrapfigure}
Under \Cref{defn:epsilon-delta-pointwise-indistinguishable}, the two
distributions are $(\epsilon, \delta)$-pointwise indistinguishable if the
shaded area in the interval $[-\epsilon, \epsilon]$ is at least $1-\delta$.
Using this visual criterion, the definition of
$(\epsilon, \delta, \alpha)$-random differential privacy says that, if we
repeatedly sample similar inputs from $\mathcal{I}$ and inspect the
corresponding graph of the privacy loss random variable, then with probability
at least $1-\alpha$, we will see a graph whose shaded area is at least
$1-\delta$. The example graphs show seven privacy loss random variable
distributions whose shaded area in the interval $[-\epsilon, \epsilon]$ is at
least $1-\delta$ (colored in blue), and two whose shaded area is less than
$1-\delta$ (colored in red). Visually, if we extend this grid of graphs with
privacy loss random variables derived from more and more sampled similar inputs
from $\mathcal{I}$, then the parameter $\alpha$ bounds the fraction of red
graphs in the entire grid.

\endgroup

We say a program $f$ is $\epsilon$-\emph{well-behaved} if $f$ is $(\epsilon,
0)$-differentially private and if $f$'s path conditions are exactly
necessary and
sufficient for proving $(\epsilon, 0)$-differential privacy through the
pointwise proof technique.
\begin{lem}
\label{lem:well-behaved}
If a program $f$ is $\epsilon$-well-behaved, then $f$ is never rejected
by \fuzzidp's testing framework when tested with any $\epsilon' \geq \epsilon$.
%\bcp{Vague wording is confusing here: f may
%  well be rejected if the tester is told to check that it is $\epsilon' <
%  \epsilon$ very well behaved!}
\end{lem}
\noindent The proof of \Cref{lem:well-behaved} can be found in \Cref{ap:proof-well-behaved}.

We might hope that all $(\epsilon, 0)$-differentially private programs are
$\epsilon$-well-behaved, but this does not hold in general
because \fuzzidp assumes that proofs of differential privacy for these programs
have a particular structure (\Cref{sec:testing-differential-privacy}): the path
conditions for these programs must be the \emph{neccessary}
and \emph{sufficient} conditions for its privacy
properties. In \Cref{sec:evaluation} we will see the ReportNoisyMaxWithGap
algorithm, whose optimal $\epsilon$ is rejected by \fuzzidp because its path
conditions are sufficient but not necessary. (\fuzzidp does accept
ReportNoisyMaxWithGap with a non-optimal $\epsilon$, for which its path
conditions are both necessary and sufficient.)
%\bcp{It is rejected
%with its minimum $\epsilon$ but accepted with a non-optimal one, right?  This is
%not nearly so worrying.  (Also, by the way, we do say in at least one place --
%unless you've fixed it by now -- that all the benchmark programs are accepted!)}

Conversely, assume a fixed distribution $\mathcal{I}$ of similar inputs and a
program $f$. We say that $f$ is $(\epsilon, \delta, \alpha)$-\emph{ill-behaved}
% \BCP{This doesn't make sense to me: here we're talking about a particular
%   $\delta, \alpha$, but the next thing we do is to quantify over them...}
% \hzh{We are quantifying over $\delta_f$ and $\alpha_f$. The $\delta$ and
% $\alpha$ in ill-behaved are thresholds that lower bound all valid
% $\delta_f$ and $\alpha_f$ once we fix $\epsilon$. }
% BCP: Ah, got it.
if, given fixed $\epsilon$, all valid random differential privacy parameters
$(\epsilon, \delta_f, \alpha_f)$ for $f$
satisfy $\delta_f > \delta$ and $\alpha_f > \alpha$. Intuitively, an
$(\epsilon, \delta, \alpha)$-ill-behaved program has a high probability of two
kinds of failure---its ``catastrophic failure'' probability is at least $\delta$
when executed on ``good'' similar inputs from $\mathcal{I}$, and there is at
least an $\alpha$ probability of draws from $\mathcal{I}$ yielding ``bad''
similar inputs such that, when $f$ runs on these inputs, there is a greater
than $\delta$ chance that $f$ will induce more privacy cost than $\epsilon$.
%\ar{This isn't quite right --- we don't know that $f$ provides no privacy guarantees at all on bad inputs. I would just say that we know that with probability at least $\alpha$ over inputs, the probability (over the randomness of the mechanism) that the privacy loss exceeds $\epsilon$ is at least $\delta$. Maybe point out that this is useful for testing because we know that we don't need more than roughly $1/\alpha$ many sampled instances before we find one on which there will be no coupling proof after we sample roughly $1/\delta$ many traces.}
\begin{thm}
\label{thm:sample-complexity}
Given a fixed distribution $\mathcal{I}$ over similar inputs, a positive
integer $k$, an $(\epsilon+k\omega,\delta,\alpha)$-ill-behaved program $f$, and
a positive value $\theta$, if
\begin{enumerate}
\item $f$ makes at most $k$ calls to the Laplace sampling instructions in one
execution,
\item $f$ has at most $n$ output buckets
(\Cref{sec:testing-differential-privacy}), and
\item \fuzzidp failed to reject $f$, because \fuzzidp discovered $\mathit{shift}_i$ values
(used in \cref{eq:generic-shift}) that are valid for the sampled execution
traces,
\end{enumerate}
then, as long as we had run at least $d$ tests with independently sampled
inputs and with at least $m$ sampled traces in each test, the probability of
such
a failure invalidating the claim of $(\epsilon, 0)$-differential privacy for
$f$ is at most $e^{-d(\theta + \alpha)}$, as long as
$m \geq \frac{1}{\delta}(\theta + nk\ln 2 +  nk\ln\frac{C_2 - C_1}{\omega})$ where $C_1
= \min_i \mathit{shift}_i$ and $C_2 = \max_i \mathit{shift}_i$.
%\bcp{Whew.
%  There were some weird grammatical things in that, which I tried to fix,
%  but I am not sure I got it right.  And the whole thing is kind of a
%  monster.}
% hzh: looks good to me
\end{thm}
%\bcp{Proof?  (I suspect we don't have room for it here, but at least we need
%a pointer to where readers can find it.)}
\noindent The proof can be found in \Cref{ap:proof-of-sample-complexity}.

In practice, the $\mathit{shift}_i$ values are bounded by machine limits. Even
if we take $C_1$ as the smallest double-precision floating point number, $C_2$
as the largest, and $\omega$ as the smallest gap between two double-precision floats, the factor $\ln \frac{|C_2-C_1|}{\omega}
= \ln \frac{1.798 \times 10^{308} \times 2}{2^{-52}}$ is smaller than
$747$, a small requirement on the number of sampled traces even in this extreme case.
%\bcp{Is this a little or a lot??}

Note that there is a non-zero gap
%\bcp{don't understand what it means to have
%a numerical gap between a claim and a degree...!}
of $k\omega$ between the tested privacy level $\epsilon$
 for $f$ and the level $(\epsilon
+ k\omega, \delta, \alpha)$ to which $f$ is ill behaved. This means \fuzzidp is
only guaranteed to catch bugs with high probability if $f$'s behavior differs
enough from the claimed levels of differential privacy. However, since $\omega$
is the granularity of the discretized domain, the value $k\omega$ is typically
very small. (For example, the gap between two double-precision floating point
numbers in the interval $[0, 1)$ is $2^{-52}$.) Conversely, if $f$'s behavior is
not very far from its claimed level of differential privacy, then we cannot give
any guarantees about the probability of falsely accepting $f$.

Although we presented the testing strategy through the example
algorithm \lstinline|rnm|, the testing framework requires no special input
particular to \lstinline|rnm|. In fact, we can use the testing strategy
described here to check the differential privacy property of many other
algorithms, as long as these algorithms' differential privacy proofs follow the
general template listed in steps 1 to 4 and their program control flow
conditions are neccesary and sufficient for differential
privacy. In \Cref{ap:detailed-evaluation}, we describe a variety of algorithms
tested using this strategy; we also present the ReportNoisyMaxWithGap algorithm,
whose optimal privacy cost cannot be established using this strategy, though it
can still be validated as differentially private with a non-optimal $\epsilon$.
In \Cref{sec:evaluation}, we discuss the details of a practical workflow that
applies \fuzzidp to develop, test, and integrate core differential privacy
mechanisms with an existing software system designed for the 2020 US Census.

%\bcp{we say elsewhere that this result doesn't actually give us practical
%  confidence in the absence of false positives; shouldn't we also say it
%  here in this section (e.g., maybe at the top, with a reference to a longer
%  discussion later)?}

\section{Evaluation}
\label{sec:evaluation}
%\bcp{Perhaps we should open the section with a clear statement of
%  ``questions we seek to answer...''?}
We seek to answer the following questions:
\begin{enumerate}
\item[1.] How expressive is \fuzzidp's testing strategy?
\item[2.] Can \fuzzidp assist implementations of real-world systems that use
differential privacy?
\end{enumerate}
To answer these questions, we first used \fuzzidp to distinguish private and non-private
variants of 10 differential privacy benchmark algorithms from the
literature.  Second, we
used \fuzzidp in a practical workflow to re-implement and test the core
differential privacy mechanism from the Disclosure Avoidance System (DAS)
for the 2020 US Census~\cite{10.12688/gatesopenres.13089.1}.
To save space, we concentrate on the latter experiment and give just a
brief summary of the former; further details can be found in
\Cref{ap:evaluation-benchmark-algorithms}.

\subsection{Benchmark Algorithms (Summary)}
\label{subsec:benchmark}

\newcommand{\cmarkmaybe}{{\color{dkgray}\cmark$^?$}}
\newcommand{\xmarkmaybe}{{\color{dkgray}\xmark$^?$}}

\begin{figure}[t]
\begin{tabular}{cccccccccccc}
Framework & & \lstinline|nc| & \lstinline|nm| & \lstinline|ns| & \lstinline|ps| & \lstinline|pt| & \lstinline|rnm| & \lstinline|rnmGap| & \lstinline|ss| & \lstinline|sv| & \lstinline|svGap| \\
\toprule
\fuzzidp & Correct  & \cmark                 & \cmark                 & \cmark                 & \cmark & \cmark & \cmark & \cmark$^*$ & \cmark & \cmark & \cmark\\
         & Buggy    & \cmark                 & \cmark                 & \cmark                 & \cmark & \cmark & \cmark & \cmark & \cmark & \cmark & \cmark \\
\toprule
LightDP  &          & \cmarkmaybe & \cmarkmaybe & \cmarkmaybe & \cmark & \xmark & \xmark & \xmark & \cmark & \cmark & \cmarkmaybe \\
StatDP   &          & \cmark      & \cmark      & \cmark      & \cmark & \xmark & \cmark & \xmark & \cmark & \cmark & \cmark \\
ShadowDP &          & \cmark      & \cmark      & \cmark                 & \cmark & \xmark & \cmark & \xmark & \cmark & \cmark & \cmark \\
Proof Synthesis &   & \cmarkmaybe & \cmarkmaybe & \cmark                 & \cmark & \xmarkmaybe & \cmark & \cmarkmaybe & \cmark & \cmark & \cmarkmaybe \\
DP-Finder &          & \cmark & \cmark & \cmark                 & \cmark & \xmark & \cmark & \xmark & \cmark & \cmark & \cmark
\\
\toprule
\end{tabular}
\caption[Caption of benchmark table]{\fuzzidp test results and coverage comparison on benchmark algorithms with other frameworks.}
\label{table:test-results}
\end{figure}

We used \fuzzidp to implement a suite of benchmark algorithms from the literature. For
each algorithm, we built both a correct implementation and several non-differentially-private variants. We expected \fuzzidp to accept the correct implementation
according to \Cref{lem:well-behaved} and to detect and reject all
non-differentially-private variants. The results are shown
in \Cref{table:test-results}.  We
place a \cmark{} in the ``Correct'' row if \fuzzidp accepts the correct
implementation under the algorithm's optimal privacy cost, and we  put a \cmark{}
in the ``Buggy'' row if \fuzzidp rejects all non-differentially private variants
of this algorithm.  For the ReportNoisyMaxWithGap (\lstinline|rnmGap|)
algorithm, we write \cmark$^*$ to indicate
that \fuzzidp does not
accept its correct implementation with the optimal privacy cost but does accept
with twice the optimal privacy cost. Details can be found
in \Cref{paragraph:rnmgap}.

We also compare the coverage over these ten benchmark
algorithms between \fuzzidp and related frameworks
in \Cref{table:test-results}. For each related framework, \cmark{} indicates
that the framework in the corresponding row has successfully analyzed the
algorithm in the corresponding column. A \xmark{} represents that the framework
in the corresponding row cannot be used to test or verify the algorithm in the
corresponding column. A gray \cmarkmaybe indicates that the authors of this
framework have not
presented an evaluation of the algorithm in the corresponding column either in
its publication or its released software artifact, but that we believe the
framework
has enough expressive power to handle this algorithm. Similarly, a
gray \xmarkmaybe indicates our belief that the algorithm in the corresponding
column is beyond the capabilities of the framework.

The table shows that \fuzzidp correctly accepts private implementations and
rejects faulty variants of all previous studied benchmark
algorithms. Additionally, it is the first framework able to distinguish
between correct and faulty variants of PrivTree.

PrivTree is challenging for automatic verification of differential privacy
for at least
two reasons. The first is that PrivTree terminates
probabilistically, i.e., the probability of PrivTree \emph{not} terminating
after $n$ iterations of its main loop diminishes as $n$ increases. The second
reason is that the privacy analysis used in the PrivTree's privacy
proof involves intermediate privacy costs that depend on input
values \cite{Zhang:2017:LTA:3093333.3009884}.

The first characteristic poses issues for static analyses (including \fuzzidp's symbolic
interpreter), as we cannot
statically know how many iterations PrivTree will run.  Fortunately, as
discussed in \Cref{sec:testing-differential-privacy}, our symbolic interpreter only needs to produce trees that match those observed in the instrumented
execution. The trees produced by PrivTree contain strictly more nodes as loop iteration counts increase. Thus, our symbolic
interpreter can stop searching for
matching trees once it
realizes that all future iterations will produces trees that cannot match those
observed in the instrumented executions.

The second characteristic is a serious issue for tools aimed at automatically
generating {\em proofs} of differential privacy. Since such tools need to reason
over all possible input values, intermediate privacy costs that
depend on inputs must be represented by expressions over these unknown
inputs. For PrivTree, these intermediate privacy cost expressions involve
non-linear arithmetic, an undecidable theory that can only be solved in a
best-effort way by SMT solvers.

By contrast, \fuzzidp's testing framework chooses a pair
of \emph{concrete} input values and evaluates PrivTree over these inputs. This
allows the testing framework to represent intermediate privacy
costs with much simpler symbolic expressions.

The two most challenging algorithms for existing validation and verification
frameworks are PrivTree and ReportNoisyMaxWithGap. StatDP and DP-Finder cannot
process the complex output type of PrivTree---tree data structures, since both
of these tools rely on heuristics that are designed to detect DP violations on
numerical outputs. LightDP and ShadowDP, the two type-system based tools, fail
at PrivTree due to the unbounded probabilistically terminating main loop. We
believe PrivTree is out-of-scope for the proof synthesis framework for the same
reason---the proof synthesis framework would not be able to analyze an unbounded
probabilistically terminating main loop.

We believe StatDP and DP-Finder's implementations could be improved with
additional heuristics to handle the output datatype of ReportNoisyMaxWithGap,
while LightDP and ShadowDP require changes to their type systems. The proof
synthesis artifact is unavailable, but since it is based on apRHL, and apRHL is
expressive enough to prove differential privacy for ReportNoisyMaxWithGap, we
believe the proof synthesis framework can handle ReportNoisyMaxWithGap.

\subsection{Disclosure Avoidance System}
\label{subsec:das}

%\bcp{We can explain more clearly why this is an interesting / revealing /
%  useful experiment for us to do.  (In particular, the fact that it is a
%  case study in how we might do things like this for other real-world DP
%  systems.)}

%\bcp{the opening is WAY too abrupt.  Indeed, for non-US readers, we need to
%  step back and explain what the census itself is.}
Every ten years, the US Census Bureau conducts a national survey to count
the total population in
the United States. This survey, refered to as the
Decennial Census, provides critical information for the Federal Government to
adjust allocation of funds, as well as representation in the US House of
Representatives, where each state gets a number of delegates proportional to
its population. For the 2020 Census, the US Census Bureau
developed an open-source Disclosure Avoidance System (DAS) to
aggregate raw survey data into population counts. DAS applies differential
privacy to protect the privacy of survey participants.

DAS aims to produce differentially private population counts for each
geographical unit within each of the six geographical levels in the US:
the whole nation, individual states,
counties, ``census tracts,'' ``block groups,'' and
single city blocks~\cite{10.12688/gatesopenres.13089.1}. This process would be
straightforward if the only requirement were differential privacy: just count the
population in each geographical unit and add appropriately sampled noise to
each count. However, a census report produced through this idealized process
would
contain inconsistent counts due to the added noise: for example, the population
count in a state might well be different from the sum of the population
counts from all
counties within the state, and there might even be negative counts for some
geographical units where the precise count before adding noise was
small. The Census Bureau has a list of data requirements that rules out
these inconsistencies, and the final report produced by DAS must satisfy these
requirements~\cite{10.12688/gatesopenres.13089.1}.

To address these issues, DAS applies a so-called ``TopDown Algorithm.'' The
TopDown Algorithm consists of 2 phases. The first phase calculates precise (and
secret) counts for all geographical units, then adds appropriately sampled noise
to produce noisy public counts.

The second phase iterates over the geographical hierarchy, ordered from coarsest
(nation) to finest (block). Each step takes the noisy counts from two adjacent
levels and perturbs them using constrained
optimization. The constrained optimization process perturbs noisy counts to
rule out inconsistencies,
%\ar{``guided by formulas'' sounds very mystical. We just mean that the optimization is just trying to minimize some notion of error subject to consistency constraints, right?},
while the optimization objective keeps the
overall perturbation of noisy counts as small as possible.

Since the outputs from the first phase already satisfy differential privacy,
the outputs from the second phase do too (because
differential privacy is robust to
postprocessing~\cite{10.1007/11681878.14}). The second phase does not introduce
any additional randomly sampled noise.

DAS is an interesting target for differential privacy testing due to
the social importance of DAS's privacy properties. Furthermore, by
applying \fuzzidp on a piece of large, real-world software artifact like DAS, we
gain insight on how \fuzzidp can assist in developing real-world software
systems that interact with sensitive data through differential privacy in the
future.

For this evaluation, we manually re-implemented, in \fuzzidp, the core privacy
mechanism that calculates the scale of noise distributions and releases the
noisy counts for each geographical unit.  We used the \fuzzidp testing
framework to check that this
mechanism is, indeed, differentially
private, and verified \fuzzidp can reject faulty variants of this mechanism. The faulty variants were edited from the correct implementation to simulate common programming mistakes---sampling noise with wrong parameters, and iterating over the input list with off-by-one errors.
%\bcp{Did we also try checking that we can detect some non-DP mutants??}.
Finally, we mechanically extracted the \fuzzidp code to Python3
code (by pretty-printing, essentially).

\Cref{fig:das-core-mechanism} shows
the \fuzzidp code that re-implements the core privacy mechanism.
\begin{figure}[t]
\begin{lstlisting}
geometricFixedSens :: Expr Int
                   -> Expr Double
                   -> Expr Double
                   -> Expr (&$\distr$& (Int, Double))
geometricFixedSens trueAnswer sens eps = do
  let alpha = fexp ((- eps) / sens)
  let prob  = 1 - alpha
  noisedAnswer <- geo trueAnswer alpha
  return (pair noisedAnswer (2 * alpha / (prob * prob)))

loopGeometricFixedSens :: Expr [(Int, (Double, Double))]
                       -> Expr (&$\distr$& [(Int, Double)])
loopGeometricFixedSens inputs = do
  (_, outputs) <- loop (inputs, nil) loopCond loopIter
  return outputs
  where loopCond (inputs, _) =
          (neg $ isNil inputs :: Expr Bool)
        loopIter (inputs, outputs) = do
          let thisInputAndTail = fromJust $ uncons inputs
          let thisInput        = fst thisInputAndTail
          let more             = snd thisInputAndTail
          let trueAnswer       = fst thisInput
          let sens             = fst (snd thisInput)
          let eps              = snd (snd thisInput)
          (noisedAnswer, variance) <-
            geometricFixedSens trueAnswer sens eps
          return $ (more, snoc outputs (pair noisedAnswer variance))
\end{lstlisting}
\caption{DAS Core Mechanism}
\label{fig:das-core-mechanism}
\end{figure}
%
\iffalse
% BCP: to unconfuse emacs syntax highlighting!
$
\fi
%
The function \lstinline|geometricFixedSens| takes a precise
count \lstinline|trueAnswer|, a parameter \lstinline|sens| that bounds the difference of \lstinline|trueAnswer| between similar inputs, and the amount of privacy budget
allocated for adding noise to this
value \lstinline|eps|. Here, \lstinline|trueAnswer| corresponds to the accurate
and secret count of population in a geographical unit. The value
of \lstinline|sens| measures
how much the precise count of this geographical unit can change between two
similar inputs of the Census data; it is determined manually by Census scientists. The value of \lstinline|eps| is also
determined by Census scientists, to provide a suitable level of privacy
protection. From \lstinline|sens| and \lstinline|eps|, we can calculate the
appropriate $\alpha$ parameter of the two-sided geometric distribution and add
sampled noise to \lstinline|trueAnswer| to produce a differentially private
noisy answer.

The function \lstinline|loopGeometricFixedSens| takes a list of input
values in the form
$$(\mathtt{trueAnswer}_1, (\mathtt{sens}_1, \mathtt{eps}_1))\,,\, \dots\,,\, (\mathtt{trueAnswer}_n, (\mathtt{sens}_n, \mathtt{eps}_n))$$ and creates an
output list of the same size that contains the noisy answers for each
$\mathtt{trueAnswer}_i$ in the list. The privacy guarantee of this procedure is
that it is $(\sum_i \mathtt{eps}_i, 0)$-differentially private.

%\bcp{Explain why this is a reasonable way to do the experiment, and explain
%  more clearly what the experiment is.  (In
%  particular, ``six hours'' seems to come out of nowhere, and I'm not sure
%  what ``repeatly test'' means, exactly.)}
  We try to experimentally detect violations of differential privacy from \lstinline|loopGeometricFixedSens| by randomly generating similar
list inputs, running \fuzzidp's testing framework with the generated lists, and checking that each test reports no violations. We repeat this entire procedure in a non-terminating loop running on a cloud virtual machine that stores all test logs and raises alarms for any test failure. The size of the input lists increase with each passing test, up to
100.

We observed that the tests occasionally fail due to a tiny over-use
of the $\epsilon$-privacy budget---in the range $10^{-12}$ to $10^{-15}$.
We believe this was caused by rounding errors from the floating point operations
that calculate noise distribution parameters
in \lstinline|geometricFixedSens|. We also observed that if we relax the privacy
parameter to $\sum_i \mathtt{eps}_i + 10^{-12}$, then all of our test cases
passed.
%\bcp{Don't understand the grammatical tenses here.  Do we just mean
%  ``all our tests passed''?  Or something more subtle?  In general, this
%  section pings back and forth between tenses in a way that makes it hard to
%understand: please go through and make it consistent}.
This suggests that
the total privacy cost incurred by
running \lstinline|loopGeometricFixedSens| is a small value plus the sum of
intended $\epsilon_i$ values, due to rounding in floating point operations.

Indeed, we expect the original version of DAS to also have this property, since its
Python3 implementation also uses floating-point arithmetic. We
confirmed
%\bcp{We expect this, or we've confirmed it?}\hzh{Yes, we confirmed it by intercepting inputs to the unmodified core DAS mechanism running on the 1940s dataset}
this
conjecture by intercepting the raw inputs to the core privacy mechanism in the
original version of DAS, converting these inputs into 128-bit floating point
numeric representations, which preserves much more precision than the 64-bit
floating point numbers used in DAS, and calculating the total privacy cost in
128-bit floats. We then compared this higher-precision total privacy cost with
the total privacy cost reported by the original version of DAS. This comparison
reveals that the total privacy cost is around $1.03 \times 10^{-11}$ more than
the reported total privacy cost. Of course, because the extra privacy cost is
extremely small, it does not significantly degrade the privacy protections
provided by DAS.

Finally, to confirm that our re-implemented privacy mechanism behaves the same as the original
version, we extract \lstinline|loopGeometricFixedSens| to Python3 and replace
the original core mechanism with the extracted code (shown
in \Cref{ap:extracted-code}). We then apply the following test setup to compare
the behavior between two versions of DAS. For each trial:
\begin{enumerate}
\item Run both versions of DAS 500 times.
\item Assume that both groups of outputs come from the same
distribution (since we assume both versions of DAS exhibit identical behavior);
run statistical test to check if there is evidence to reject this assumption.
\item Record the $p$-value from the hypothesis test.
\end{enumerate}
If our null hypothesis---that both versions of DAS behave identically---is true,
then we should observe that the recorded $p$-values follow a uniform
distribution on the interval $[0,
1]$~\cite{doi:10.1198/000313008X332421}. Accordingly, we
perform one final hypothesis test on the recorded $p$-values to search for
evidence that suggests otherwise.

To run DAS, we need census data as inputs. The US Census Bureau tested DAS's
functionality using 1940 Census
data~\cite{rugglesfloodgoekengrovermeyerpacassobek}. With the 1940 Census data,
each DAS run takes roughly 6 hours on our test machine. Since we need to perform
$500$ runs on each version of DAS per trial, and perform many trials to record
enough $p$-values, we cannot afford to run DAS on the full 1940 Census dataset.

Instead, we subsample around $1$ percent of the 1940 Census dataset and perform our
trials over this smaller dataset. On our subset of the 1940 Census data, each
run takes around $10$ minutes to finish and produces a vector of $287509$
counts for the geographical units contained in the smaller dataset. We also
parallelize the trials with $12$ machines to speed up the entire test process.

Since our null hypothesis is that the two versions of DAS produces the same output
distribution, we need a statistical test that can invalidate our null hypothesis
based on the observations of the output vectors. We use multivariate permutation
testing~\cite{chung2016multivariate} for this task. The multivariate permutation
test takes two groups of samples as inputs. In our case, the two groups are each
500 vectors, one produced by our modified version of DAS and the other produced
by the original version of DAS. The test randomly swaps vectors between
these two groups and compares a test statistic derived from the difference of sample mean vectors
%\ar{What are ``the'' summary statistics?}
from them before and after
swapping. Intuitively, if both groups of samples truly come from the same
distribution, then the test statistic should not change much due to
swapping.

\begingroup
\setlength{\columnsep}{2pt}%
\setlength{\intextsep}{0pt}%
\begin{wrapfigure}{r}{0.5\linewidth}
\centering
\includegraphics[width=0.9\linewidth]{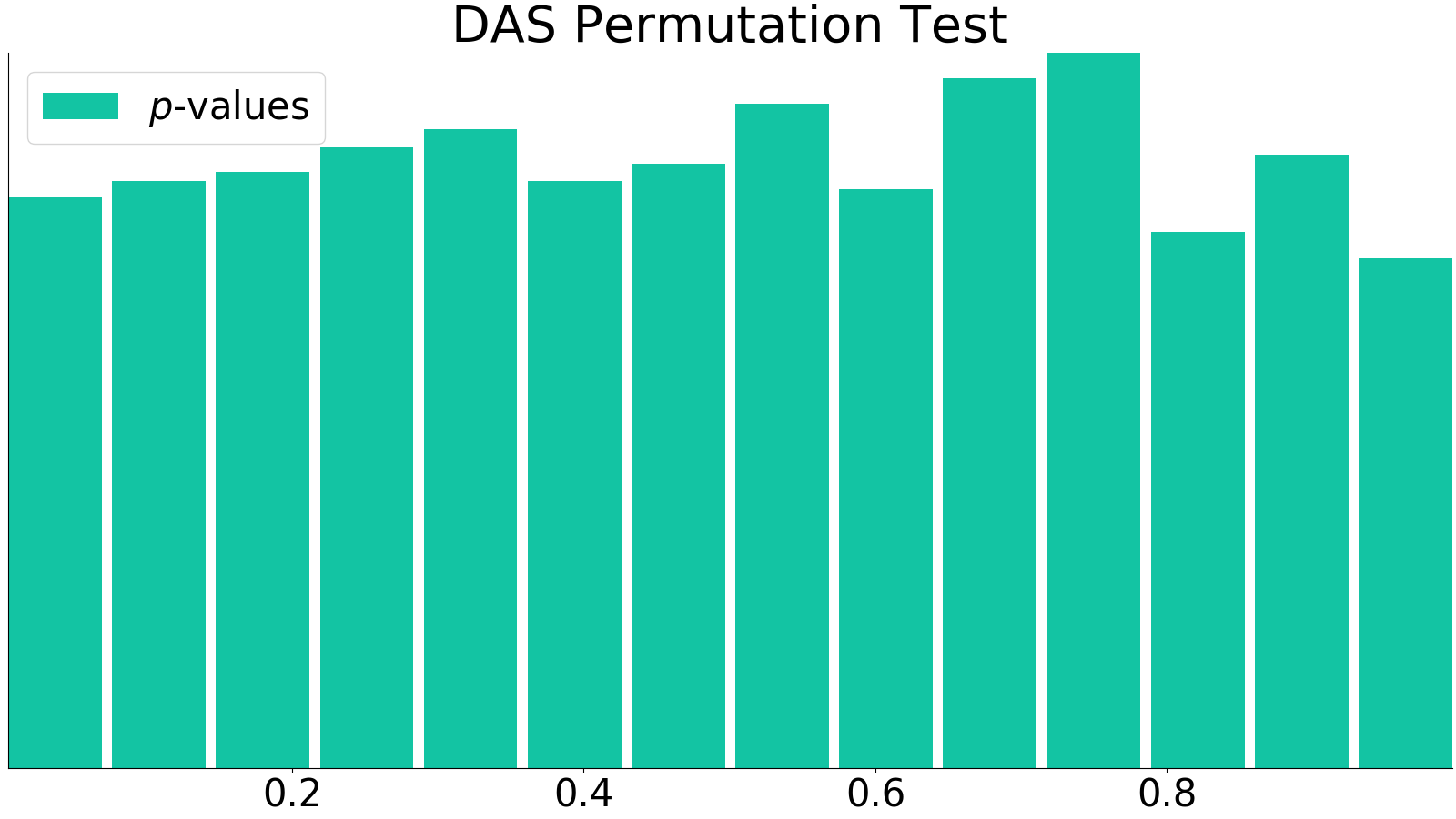}
\label{fig:p-values}
\end{wrapfigure}
Each run of the permutation test produces a $p$-value. Tests that consistently
produce very small $p$-values are evidence invalidating our null hypothesis that
two versions of DAS have identical behavior. Here, we are checking for the lack
of such evidence: when our assumption is indeed true, the $p$-values are
samples drawn uniformly at random in the interval $[0, 1]$. We perform a final
Kolmogorov-Smirnov test~\cite{doi:10.1080/01621459.1951.10500769} to check if there is evidence for
rejecting the hypothesis that $p$-values are drawn from a uniform
distribution. This final test produces a $p$-value of $0.68$, signaling a lack
of evidence to reject the hypothesis that the recorded $p$-values are sampled
from a uniform distribution. We can also plot a histogram of the observed
$p$-values (\Cref{fig:p-values}), which visually shows the collected
$p$-values. These
results produce no evidence for rejecting our null hypothesis
that the two versions of DAS indeed behave identically.

\endgroup

%\begin{figure}[t]
%\begin{center}
%\includegraphics[width=0.6\textwidth]{permtest}
%\end{center}
%\caption{Recorded $p$-values for comparison between two versions of DAS}
%\label{fig:p-values}
%\end{figure}

To summarize our evaluation of \fuzzidp on the DAS workflow: we successfully
re-implemented the privacy mechanism used by DAS, tested its privacy properties
using \fuzzidp's testing framework, extracted the \fuzzidp implementation into
Python3 code, re-integreted this extracted privacy mechanism with the rest
of DAS, and confirmed statistically that this modified version of DAS
behaves the same as the original. This case study demonstrates that we can develop
and test core differential privacy mechanisms in \fuzzidp and then integrate
these core procedures with large software systems through mechanized code
extraction.

\section{Optimizations}
\label{sec:optimizations}
%\vspace{-0.6em}
\paragraph{Speeding up symbolic execution}
We described a program transformation, \lstinline|streamline|, that
turns \fuzzidp programs into straight-line programs in \Cref{subsec:implementation}. This
transformation produces exponentially many straight-line programs that each
need to be analyzed, placing a bottleneck on the size of inputs we
can test. To mitigate this blowup in some cases, we implemented an optimized
symbolic interpreter that does not require \lstinline|streamline| and
instead applies a type-driven state-merging algorithm developed
by \citet{Torlak:2014:LSV:2666356.2594340}\ifextended{} for a LISP-based language with
integers, booleans, and
lists. We expanded the state-merging algorithm to all base types and data
structures used in \fuzzidp, including lists, tuples, and maps. \fi\xspace This
massively speeds up \fuzzidp's symbolic execution and allows us to
scale the generation of symbolic formulas to much larger input sizes. However,
these formulas are more complex than those produced by
the simpler method, and solving them is
likely NP-complete. After comparing the end-to-end testing time with and without
this optimization, we observed that the only algorithms that saw a speedup are
the ones whose intermediate states merge well (such as ReportNoisyMax), while
algorithms such as SparseVector, whose intermediate states do not merge well,
performed worse than with the original approach. \fuzzidp's testing framework
supports both versions, so that users may take advantage of the
state-merging optimizations when appropriate.
%\bcp{So... where did that
%  leave us?  Do we use this optimization, or not use it, or turn it on sometimes?}
%\vspace{-0.6em}
\paragraph{Bucketing \lstinline|Double| results}
As described in \Cref{sec:testing-differential-privacy}, the testing process
involves a ``bucketing'' step that groups sampled traces with the same
outputs. This step is easy for algorithms that yield a small number of different
outputs, as we only need to perform equality tests to group the sampled
traces. However, \fuzzidp is not limited to such algorithms. For example,
SmartSum and SparseVectorGap both yield \lstinline|Double|s, which are computed
using values sampled from Laplace distributions. It is highly unlikely that {\em
any} two runs of SmartSum or SparseVectorGap will produce the same Laplace
samples, even if they follow the same control flow path. So we cannot simply use
equality tests to bucket the outputs.

One solution is to restrict the output types of algorithms so that they only
contain a small number of possible values, but this severely limits the kinds of
algorithms that can be tested with \fuzzidp.

Instead, \fuzzidp chooses a heuristic that trades off some completeness for
allowing programmers to test algorithms that may return
sampled \lstinline|Double| values. At test time, \fuzzidp's instrumented
interpreter attaches a distribution provenance to each sampled value and to
results of arithmetic expressions that involve sample values.

For example, if $x_1$ and $x_2$ are two independent samples from the Laplace
distribution with center $0$ and width $1$, then $x_1$ and $x_2$ have
distribution provenance $\lap{0}{1}^1$ and $\lap{0}{1}^2$, and an expression
such as $x_1 \cdot x_2$ has distribution provenance
$\lap{0}{1}^1 \cdot \lap{0}{1}^2$. The superscript allows us to distinguish the
distribution provenance of $x_1 \cdot x_2$ from $x_1 \cdot x_1$ and thus
recognize that these are two different distributions.

\fuzzidp's testing framework then buckets output sample values based on the
equality between distribution provenance structures when the output values are not equal due to independent sampling.
%\bcp{Don't understand ``when these sample values are
%not equal to each other''---which sample values??}
Since most algorithms
that do return sampled values
only sample from a handful of possible output distributions, this heuristic
significantly cuts down the number of output buckets for such algorithms.

This heuristic does not sacrifice soundness with respect to the \textsc{PW-Eq}
proof rule. \textsc{PW-Eq} allows us to construct one pointwise proof for each
pair of equal output values. Here, \fuzzidp's heuristic still adheres to this
quota; indeed, it goes a step further by posing an even more stringent quota that only
allows one pointwise proof for each pair of equal output distributions. In our
evaluation, this heuristic allows us to test several benchmark algorithms that
return sampled \lstinline|Double| values. Attaching distribution provenance values to sample values introduces some interpretation overhead, but this overhead is not a bottleneck in \fuzzidp's testing performance in our evaluation.

\section{Limitations}
\label{sec:limitation}
%\iflater
%\paragraph{Expressiveness}\bcp{I[m not sure I'd count this as a limitation.
%To be sure, there are things we don't cover, but that's true of every static
%analysis.  I propose moving this discussion to the last section.}
%Although \fuzzidp's testing framework is built around a general proof framework,
%as our experiments on ReportNoisyMaxWithGap show, there certainly exist
%algorithms currently out of \fuzzidp's reach.\bcp{Moreover, is this proof
%framework believed to be complete?  I don't think so.} \fuzzidp only uses the
%shift relation to connect sampled traces with symbolic samples---a deliberate
%choice for simplifying the testing
%framework. \citet{Albarghouthi:2017:SCP:3177123.3158146}'s proof synthesis
%framework describes several other more sophisticated coupling techniques that
%may increase \fuzzidp's expressiveness.
%\iflater
%\bcp{Needed?}
%We plan to study these choices in the
%design space in future work.
%\fi
%\paragraph{Scalability}\bcp{Ditto.} \fuzzidp uses Z3 to solve an NP-complete problem in order to
%check differential privacy. The fundamental computational hardness severely
%limits \fuzzidp's testing efficiency on larger inputs. A possible direction that
%avoids this problem, is to ask programmers to write down a functional
%specification of the coupling relation between samples in two related runs of a
%program, and instead of requiring Z3 to synthesize a shift vector, we check
%whether the specification satisfies all of the requirements of a valid
%privacy proof.
%% We also plan on supporting these features
%% in \fuzzidp in future work.
%\fi
\paragraph{Gap from Theoretical Guarantees} Due to scaling
issues, our current testing framework does not allow us to run tests large
enough to give meaningful $(\epsilon, \delta, \alpha)$-random
differential privacy guarantees through \Cref{thm:sample-complexity}. For
example, to achieve guarantees with $\delta = 10^{-5}$,
by \Cref{thm:sample-complexity} we know we need
at least $10^5$ samples. Our current evaluation uses between $500$ to $5000$ sampled
traces per test iteration---i.e. the theoretical lower limit on sampled
traces is at
least 20 times larger than our current test parameter. Furthermore, the size of
the formula to be solved by Z3 grows linearly with the number of sampled traces
we use. Since the time it takes to solve these formulas grows exponentially
in terms of formula size, testing with $10^5$ sampled
traces may conservatively slow down
testing time by around $k^{20}$ for some base exponent $k$.
%\ar{Even assuming exponential scaling, weird to guess the base of the exponent.}.
Given these scaling
issues, \fuzzidp's testing
framework is more useful for catching differential privacy bugs than for
validation.
%\vspace{-1em}
\paragraph{Numerical Implementation Issues} We study \fuzzidp's testing guarantees
by assuming a countably infinite discretized domain, where consecutive
points in the discretized domain are exactly $\omega$ apart. But
\fuzzidp's implementation uses double-precision floating point
numbers. This mismatch is known to lead to privacy
leaks
\cite{DBLP:journals/corr/abs-1904-12773,Mironov:2012:SLS:2382196.2382264}. It
can be remedied by using fixed-precision
numbers.
%\vspace{-1em}
\paragraph{Improving Symbolic Execution on Probabilistically Terminating Programs}
PrivTree is the only probabilistically terminating program we have tested in the
evaluation. To avoid infinitely unrolling the main loop of PrivTree during its
symbolic execution, we explicitly \lstinline|abort| the computation when the
outputs from symbolic execution cannot possibly be matched with those observed
in instrumented execution. This is a rather ad-hoc treatment that introduces an
unnecessary \lstinline|abort| instruction in \fuzzidp. However, it is possible
to generalize the underlying principle behind our current treatment of
probabilistically terminating programs. First, the programmer needs to identify a
metric over the program's output, and must ensure that this metric is
monotonically increasing as more iterations of the program are executed. Next,
we can find the maximum value of this metric in the outputs among the
instrumented executions, and cut off a potentially infinite unrolling in the
symbolic execution when this loop metric exceeds the maximum value observed from
instrumented executions.

%\vspace{-1em}
\section{Related Work}
\label{sec:related-work}
The diverse body of related research on programming-language approaches to
differential privacy can be roughly grouped into these
categories.
\paragraph{Axiomatic Systems} Many  languages have
been designed explicitly for implementing differentially private
programs. Notable examples are:
%\bcp{figure out how to make the enumeration
%  labels come out in parens}
\begin{ENUM}
\item Fuzz \cite{Reed:2010:DMT:1863543.1863568}, a functional programming language
with a linear type system for tracking functional sensitivity and
$\epsilon$-differential privacy,
\item DFuzz \cite{Gaboardi:2013:LDT:2480359.2429113}, a dependently typed Fuzz that
allows index-refinement types for more precise tracking of
$\epsilon$-differential privacy,
\item AdaptiveFuzz \cite{Winograd-Cort:2017:FAD:3136534.3110254}, a multi-stage
functional programming language that supports Adaptive
Composition \cite{NIPS2016:6170} of differential privacy,
\item and Duet \cite{Near2019}, a functional programming language that extends Fuzz
to approximate differential privacy.
\end{ENUM}

These languages are designed with specialized type systems that internalize
differential privacy proofs of useful mechanisms. \fuzzidp is also a language
for programming differential privacy, but it is our choice to \emph{not} use a
specialized type system for differential privacy in \fuzzidp. The datatypes
in \fuzzidp are only simple types that help programmers avoid common mistakes
such as multiplying a number by a list. By only using a standard type system, we
keep \fuzzidp's testing design applicable to conventional programming languages and idioms.
%are as general as
%possible\bcp{I'm not sure about all this talk about principles: the main
%  point is that using a standard type system makes the work applicable to
%  conventional programming languages and idioms}. Maintaining this
%generality may help us port \fuzzidp's testing
%framework to real-world languages used for data analysis tasks.

\paragraph{Mechanized Proof Systems} There has been a line of work
on developing type and proof systems for the purpose of building
machine-checkable proofs for differential privacy. Examples include:
\begin{ENUM}
\item LightDP \cite{Zhang:2017:LTA:3009837.3009884}, a dependently typed language that
uses annotations to automate differential privacy proofs,
\item ShadowDP \cite{Wang:2019:PDP:3314221.3314619}, an improvement over LightDP that
handles more sophisticated algorithms, with less annotation, and in less time,
\item \aprhl \cite{Barthe:2016:PDP:2933575.2934554}, a program logic built on probability
distribution coupling theory for manual proofs of differential privacy,
\item work by \citet{Albarghouthi:2017:SCP:3177123.3158146} on a system that
borrows
from \aprhl to automatically synthesize differential privacy proofs, and
\item work by \citet{barthe2019automated} on a system that automatically
proves
and disproves differential privacy for a language restricted to finite
domains.
\end{ENUM}

The goal of these systems is to mechanically certify privacy of small programs
with complex proofs. They all rely on general typing (or proof) rules that can
capture the proofs of mechanisms such as ReportNoisyMax and
SparseVector. LightDP, ShadowDP,
and \citet{Albarghouthi:2017:SCP:3177123.3158146}'s work all perform static
analysis with the help of a solver; for privacy mechanisms with intermediate
steps that depend on input values (e.g., PrivTree), the underlying solver would
likely yield inconclusive results due to arithmetic
complexities \cite{Zhang:2017:LTA:3009837.3009884}. Furthermore, these systems
often carry the burden of proving termination of the program under
analysis. When faced with probabilistically terminating programs
(e.g. PrivTree), they fail to produce useful analysis.
\fuzzidp achieves greater expressiveness compared to these systems by considering the privacy properties of particular runs of a program, rather than trying to prove the program's privacy property. This testing-based approach
avoids the arithmetic challenge described
by \cite{Zhang:2017:LTA:3009837.3009884}, since much of the arithmetic
complexity is evaluated away early\ifextended{} in the process\fi, and
allows \fuzzidp to gracefully test probabilistically-terminating programs by
limiting its symbolic exploration through instrumentation traces.
\ifextended
Although \fuzzidp does not produce proofs, we
can bound the probability of failing to reject faulty programs
(\Cref{thm:sample-complexity}).
\fi

\citet{barthe2019automated}'s work is unique in that the authors restrict the problem of
automatic verification \ifextended for differential privacy \fi to a finite domain. In their
model language, all data types contain a finite number of elements. This
restriction allows the authors to craft complete decision procedures that prove
or disprove differential privacy. \fuzzidp does not restrict its datatypes to
finite domains, and its testing framework is not complete. Supporting programs
with unusual differential privacy proof structures remains an important
direction of our future research.

Since \fuzzidp tests a program for differential privacy, and {\em cannot prove}
differential privacy. \fuzzidp by itself may not be sufficient for critical
applications. When absolute guarantee of differential privacy is required, a
more manual verification with a mechanized proof system is still necessary.

\paragraph{Statistical Testing} studies how to
invalidate hypotheses about probability distributions; its techniques have also
been applied to detect violations of differential privacy. Representative work
includes:
\begin{ENUM}
\item StatDP \cite{Ding:2018:DVD:3243734.3243818}, a framework for statistical testing
of differential privacy,
\item DP-Finder \cite{10.1145/3243734.3243863}, a framework that detects violations of
differential privacy through code transformation, a careful sampling technique
and objective optimization, and
\item work by \citet{wilson2019differentially}, a SQL toolkit for differential privacy.
\end{ENUM}

\fuzzidp is very similar to StatDP and DP-Finder in their goal of automatic testing
of differential privacy, but they are very different in how they achieve this
goal.

StatDP repeatedly runs the program under test, constructs two histograms
approximating the two output distributions on similar inputs, and compares these
two histograms using statistical tests to detect violations of differential
privacy.

DP-Finder applies a novel sampling technique to construct a formula that
approximates the privacy loss random variable, and then infers a lower bound of
$\epsilon$ through objective optimization on this approximation formula.

\fuzzidp also repeatedly runs a program under test on one of the two similar
inputs, using an instrumented interpreter that collects traces. Both StatDP and
DP-Finder place restrictions on the shape of outputs from programs under tests,
because both frameworks apply heuristics to detect output events that likely
indicates a violation of differential privacy. \fuzzidp adapts a general proof
technique into a testing strategy, so that we only require equality tests on
outputs of the program under test.

\citet{wilson2019differentially} developed an extension to the PostgreSQL database
that checks differential privacy properties of SQL queries. They also applied
histogram-based statistical testing to validate the correctness of their
implementation of the extension.

\section{\ifextended Conclusion and \fi Future Work}
\label{sec:conclusion}
\ifextended
Real-world deployments of differential privacy require reliable and
easy-to-use tools for testing the correctness of implementations. \fuzzidp and
its testing framework give programmers a push-button mechanism that automates
testing of differential privacy. We validate \fuzzidp's effectiveness on
challenging mechanisms and formalize \fuzzidp's testing guarantees.
\fi
%\bcp{This paragraph would actually make a good future work section.  I
%  propose moving it there (and probably deleting what is there now).}
%\bcp{We need a little more text here.   For one thing, readers will be
%  confused by ``this.''  Also this seems like a very specific thing to cite
%  as our only future work.  (Though I agree that it's one of the most
%  important bits.)}

Addressing the scalability of \fuzzidp's testing framework remains the
most important avenue for
improvement. \Cref{sec:optimizations} introduced a
type-driven optimization to remove one of two significant scaling
bottlenecks; this optimization speeds up symbolic execution, but checking
satisfiability for the
resulting formulas remains a serious bottleneck. In practice, this bottleneck
prohibits \fuzzidp's testing framework to run tests large enough for validation
at a high confidence level.

To reduce this gap between theory and
implementation, we plan on improving both sides. On the theory
side, \Cref{thm:sample-complexity} only gives a very crude lower bound on the
number of sampled traces required for a given confidence level; we believe
it can be improved by more careful analysis. On the implementation side, (1)
we can develop domain specific solver heuristics for the kinds of formulas
that \fuzzidp generates, and (2) we can develop specification-based testing
for \fuzzidp, approaching validation through programmer annotations of shift
values rather than using Z3 to synthesize them.
%HZH: seems redundant
%It remains important future
%research for us to make \fuzzidp effective for validation of correct algorithms.

Other avenues for improvement also remain. In particular, we hope to (1) harden the current
implementation using fixed-precision instead of floating point
numbers, (2) improve the current ad-hoc treatment of probabilistically
terminating programs \ifextended using a general framework of loop metrics \fi as discussed
in \Cref{sec:limitation}, and (3) incorporate more sophisticated
relations on
samples to increase expressiveness.

\section{Acknowledgments}
\label{sec:ack}
We are grateful to Danfeng Zhang, Daniel Winograd-Cort, Justin Hsu, and the Penn
PLClub for discussion and comments, and we thank the anonymous reviewers for
their detailed feedback. This work was supported in part by the National Science
Foundation under grants CNS-1065060 and CNS-1513694.

%\bcp{Trim / compress?}%
%Many directions for further improvement remain. (1) We can
%incorporate more sophisticated relations on samples to increase the
%expressiveness of \fuzzidp. (2) We can develop support for manually written
%specifications of relations on samples and validate differential privacy
%through the specification. (3) We can harden the implementation using
%fixed-precision arithmetics instead of floating point numbers. (4) We can
%improve the current ad-hoc treatment of probabilistically terminating programs
%using a general framework of loop metrics as discussed
%in \Cref{sec:limitation}. (5) We can engineer a parallel implementation
%of \fuzzidp's testing framework for higher testing throughput. (6) We can
%develop solver heuristics for checking of satisfiability for the type of
%formulas generated by \fuzzidp's testing framework to achieve more efficient
%testing in the easier cases.

%\bcp{The citation for \cite{kstest} looks mangled.  Also Feldspar.  I
%  haven't checked the rest, but someone should...}

\bibliography{local}

\ifappendix
\newpage
\appendix

\section{Proof of Lemma 12}
\label{ap:proof-well-behaved}
\fuzzidp's testing framework only rejects programs for which it cannot construct
valid dual executions. Here, we show that $f$ always has a valid dual execution,
so that it is never rejected by the testing framework.

If the program $f$ is $\epsilon$-well-behaved, then we know there is a pointwise
proof for its $(\epsilon, 0)$-differential privacy. In particular, for any
similar inputs $(x_1, x_2)$ for $f$, consider a run of $f$ on $x_1$ that outputs
some output $r$. Take the application of \textsc{PW-Eq} in the privacy proof of
$f$, and inspect the applications of the \textsc{Lap-Gen} rule for that output
$r$.

This sequence of applications of \textsc{Lap-Gen} specifies a sequence of shift
values for this run of $f$ applied to $x_1$. Consider a dual execution of $f$ on
$x_2$ where the Laplace samples are constructed from this sequence of shift
values specified by the applications of \textsc{Lap-Gen}. Since $f$'s path
conditions are necessary and sufficient for proving its $(\epsilon,
0)$-differential privacy, we know the constructed Laplace samples in dual
execution must follow a control flow path that leads to the same output
$r$. So \fuzzidp would not reject this $f$. \qed

\section{Proof of Theorem 13}
\label{ap:proof-of-sample-complexity}
Given the assumptions of \Cref{thm:sample-complexity}, we know there are at most
$n$ possible output buckets. For each of output bucket, our testing process is
trying to find a valid shift vector, such that the ensemble of the $n$ shift
vectors form valid dual executions.

%HZH: this is actually only used here.
\begin{defn}
\label{defn:omega-approximation}
Given vectors $\vec{v}$ and $\vec{q}$ in $\mathbb{R}^n$, we say that $\vec{q}$
is an {\em $\omega$-approximation} of $\vec{v}$ if each coordinate $\vec{q}_i$
is in the discretized domain, and $|\vec{v}_i - \vec{q}_i| \leq \omega$.
%\bcp{So $\omega$ is both the ``horizontal'' gap between consecutive discrete
%values and the ``vertical'' maximum distance between $v$ and $q$?}\hzh{not sure
%what ``vertical'' distance means...?}\bcp{Pointwise}\hzh{yes, the idea is that
%take a vector in R^n, then snap each coordinate of this vector to
%the \omega-grid in R^n, any of the vectors on the grid that can be snapped onto
%with pointwise distance at most \omega is an \omage-approximation}
\end{defn}

However, recall that \fuzzidp's semantics assumes a discretized subset of the
reals with granularity $\omega$ for all representable numbers. Our testing
process uses Z3, which operates with rational semantics. To ensure that each
coordinate of the shift vector is in the discretized ring, we must use the
$\omega$-approximations of the rational shift vectors that Z3 finds in this
analysis. The associated privacy costs of these $\omega$-approximations will be
larger than the original prescribed $\epsilon$ privacy cost, since each
coordinate has been perturbed by up to $\omega$. Taking the cost equation
for \text{Lap-Gen} into account, the total privacy cost could increase by a
maximum of $k\omega$. So, in fact, the testing process is trying to find
evidence of $(\epsilon + k\omega, 0)$-differential privacy for the program under
test, where $k$ is the maximum size of a shift vector among the $n$ output
buckets. We will let $\epsilon'=\epsilon + k\omega$ in the remainder of the
proof.

Now, on a pair of sampled inputs, by assumption, we know $\delta$ lower-bounds
the probability that the instrumented run draws the ``bad'' sampled traces, for
which we cannot construct dual executions whose output distributions
satisfy \cref{eq:1} with $(\epsilon', 0)$. So, as long as any bad sampled trace
appears, the testing process will successfully reject this faulty program. We
are considering the opposite case, where the testing process fails to reject
this faulty program.

By assumptions, each ``good'' sampled trace appears with probability at most
$1-\delta$. So the probability that all $m$ sampled traces are good is at most
$(1-\delta)^m$. For these ``good'' sampled traces, the testing process is trying
to discover $n$ shift vectors (one for each output bucket), that are evidence
of valid dual executions.

Again, by assumptions, each shift value is in $[C_1, C_2]$ with $C_1
= \min_i \mathit{shift}_i$ and $C_2 = \max_i \mathit{shift}_i$. This bound is
only discovered \textit{ex post}, after \fuzzidp has already failed to reject
the algorithm $f$. For our theoretical analysis, we can assume \fuzzidp starts
the search of shift values in a small interval, and doubling the search space
each time it fails to discover any valid shift values until it succeeds.

So the search process coveres $\frac{C_2-C_1}{\omega} + \frac{C_2 -
C_1}{2\omega} + \frac{C_2 - C_1}{4\omega} + \dots = \frac{2(C_2-C_1)}{\omega}$
different values in each coordinate. And, since there are at most $k$
coordinates in each shift vector, and there are $n$ shift vectors in total, the
total number of possible shift vectors is
$$2^{nk}\left(\frac{C_2-C_1}{\omega}\right)^{nk}$$

If the testing process does find some set of $n$ shift vectors that lead to
valid dual executions for the $m$ sampled traces, then this test fails to reject
the program under test.
%We know that only some of the all possible shift vectors
%may fit all $m$ sampled traces, since they need to satisfy the boolean path
%conditions collected from symbolic execution.
Let $\mathbb{P}[\mathtt{FAIL}(x_1, x_2)]$ denote the total probability that the
testing process fails to reject the program under test on the sampled similar
inputs $(x_1, x_2)$. Then, we know
\begin{align*}
\mathbb{P}[\mathtt{FAIL}(x_1, x_2)] &\leq (1-\delta)^m 2^{nk}\left(\frac{C_2-C_1}{\omega}\right)^{nk} \\
&\leq e^{-\delta m} 2^{nk}\left(\frac{C_2-C_1}{\omega}\right)^{nk}.
\end{align*}
If we want $e^{-\delta m} 2^{nk} \left(\frac{C_2-C_1}{\omega}\right)^{nk} \leq
e^{-\theta}$ for some $\theta$, then we can derive a lower bound for $m$
\begin{align*}
e^{-\delta m} 2^{nk} \left(\frac{C_2-C_1}{\omega}\right)^{nk} &\leq e^{-\theta}\\
-\delta m + nk\ln 2 + \ln \left\{\left(\frac{C_2-C_1}{\omega}\right)^{nk} \right\} &\leq -\theta \\
 \theta + nk\ln 2 + nk\ln\frac{C_2-C_1}{\omega} &\leq \delta m
\end{align*}
\begin{equation}
 \frac{1}{\delta}\left(\theta + nk\ln 2 + nk\ln\frac{C_2-C_1}{\omega}\right) \leq m. \label{eq:2}
\end{equation}

Now, let's consider the probability of failing to reject $f$ if we repeat the
test with $d$ pairs of sampled inputs $$(x_1^1, x_2^1), \dots, (x_1^d, x_2^d).$$

We know that as long as $m$ satisfies the lower bound in \cref{eq:2}, then the
probability $\mathtt{P}[\mathtt{FAIL}(x_1^i, x_2^i)]$ (on any pair of similar
inputs) of failing to reject $f$ is at most $e^{-\theta}$.

By assumptions, we know that under the distribution $\mathcal{I}$ for similar
inputs, with probability at least $\alpha$, the similar inputs have no valid
proof of \cref{eq:1} under privacy budget $\epsilon$, which means our testing process can immediately reject $f$
if such a pair of inputs are generated. So we know that with probability at most
$1-\alpha$, we may get a pair $(x_1^i, x_2^i)$ that the testing process cannot
immediately reject. So, on a single pair of similar inputs, the probability of
failure to reject is at most
$$(1-\alpha)e^{-\theta} + \alpha \cdot 0 = (1-\alpha)e^{-\theta}.$$

Now, for $d$ pairs of similar inputs drawn independently at random from
$\mathcal{I}$, the probability of failure to reject any of them
$\mathbb{P}[\mathtt{FAIL}(x_1^1, x_2^1), \dots, \mathtt{FAIL}(x_1^d, x_2^d)]$ is at most
\begin{align*}
\left((1-\alpha)e^{-\theta}\right)^d &= (1-\alpha)^d e^{-d\theta}\\
&\leq e^{-\alpha d} e^{-\theta d} \\
&= e^{-d(\theta + \alpha)}.
\end{align*}

This is the bound on failing to reject a faulty $f$, as long as each test uses
more than $m$ sampled traces, where $m$ is lower bounded by \cref{eq:2}.
\qed

\section{Evaluation of Privacy Mechanisms}
\label{ap:evaluation-benchmark-algorithms}
We implemented a number of algorithms from related
work on automated verification and testing
\cite{Albarghouthi:2017:SCP:3177123.3158146,Ding:2018:DVD:3243734.3243818,Wang:2019:PDP:3314221.3314619,Zhang:2017:LTA:3009837.3009884},
as well as
the recently proposed ReportNoisyMaxWithGap and SparseVectorWithGap algorithms
from \cite{DBLP:journals/corr/abs-1904-12773} and the PrivTree algorithm.
PrivTree, in particular, is cited  by authors of
LightDP~\cite{Zhang:2017:LTA:3009837.3009884} as a challenging example.
%\iflater
%\bcp{We don't actually say that
%these are all the algorithms that are out there as far as we know, but that's a
%much stronger claim.  Is it true?  If so, we should say it, and give some
%evidence why the reader should believe us.}\hzh{I just double checked and there
%is one from LightDP that we did not cover. It's a randomized response algorithm
%based on uniform distribution, which we didn't implement in this work.}
%\bcp{OK, sounds like we need to postpone this issue for a later iteration.}
%\fi
The programs we evaluated can be found
in \Cref{ap:detailed-evaluation}.

We ran \fuzzidp's testing framework on the correct implementation of each of these
algorithms and confirmed that it is accepted.
We also implemented {\em non}-differentially private variants of each
algorithm
and measured the time \fuzzidp's testing framework takes to discover the
failure.
We used the QuickCheck \cite{Claessen:2000:QLT:351240.351266} randomized testing
library to build test input generators for
\fuzzidp's privacy testing framework. We manually wrote one generator for
each type of similarity relation introduced in \Cref{sec:review}.

We run 100 tests on correct algorithms
and verify they all pass. For catching bugs, we run only 50
tests and verify the bug is caught within 50 tests. The test count $100$ is
chosen to run the test suite for as long as possible on differentially private
programs without overloading our test environment. The test count $50$ is
empirically chosen to give the testing framework enough chances to observe
failures on all our non-differentially-private benchmarks.

The experiments were run on a virtual test machine with a 2-Core CPU clocked at
2.3GHz, and 7GB of RAM.
%We
%repeat the full test suite every 6 hours to fully utilize the allocated CPU hours
%on the virtual test machine. \BCP{I still don't see why even that last bit
%  is worth saying.}
\begin{figure}[t]
\centering
\begin{tabular}{ccllll}
Bug                      & $\epsilon$ & TTB       & Std. dev. & ITB     & Std. dev.\\
\toprule
\lstinline|nc|           & $1.0$      & $0.8s$    & $0.7s$    & $1.9$   & $1.6$\\
\lstinline|nm|           & $1.0$      & $2.3s$    & $0.5s$    & $1.0$   & $0.0$\\
\lstinline|ns|           & $1.0$      & $0.5s$    & $0.1s$    & $1.0$   & $0.0$\\
\lstinline|ps|           & $1.0$      & $2.0s$    & $0.5s$    & $1.0$   & $0.0$\\
\lstinline|pt|           & $2.58$     & $0.8s$ & $0.4s$ & $1.0$ & $0.0$\\
\lstinline|rnm|          & $2.0$      & $16.7s$   & $16.3$    & $6.4$   & $5.8$\\
\lstinline|rnmGap|       & $4.0$      & $8.9$     & $12.9s$   & $1.4$   & $0.6$\\
\lstinline|ss|           & $1.0$      & $29.8s$   & $27.7s$   & $11.0$  & $9.8$\\
\lstinline|sv3|          & $1.0$      & $8.7s$    & $3.9s$    & $1.0$   & $0.0$\\
\lstinline|sv4|          & $1.0$      & $27.1s$   & $15.8s$   & $2.3$   & $1.3$\\
\lstinline|sv5|          & $1.0$      & $7.8s$    & $8.2s$    & $1.8$   & $1.1$\\
\lstinline|sv6|          & $1.0$      & $218.8s$  & $173.4s$  & $4.1$   & $3.0$\\
\lstinline|svGap|        & $1.0$      & $1899.1s$ & $1750.2s$ & $2.1$   & $1.3$
\end{tabular}
\caption[Caption for Test Statistics]{Mean time and mean iteration to bug
  discovery, and their standard deviation on incorrect implementations. (The
  value $\epsilon=2.58$ is derived by instantiating PrivTree's privacy cost formula~\cite{Zhang:2016:PDP:2882903.2882928} with all parameters set to $1.0$.)}
%  \bcp{We're showing way too many ``significant''
%  digits!}\bcp{Good that we're showing few digits now, but I still wonder what
%  (statistically) justifies showing even one digit past the decimal place.}
%    %\BCP{And I remain very skeptical that it makes any sense to
%    %show the max value from 20 runs.  The standard deviation is what's
%    %supposed to show how much variation there is.  Another
%    %concern: what evidence do we have that 20 runs is enough to give us
%    %consistent results?}
%    \bcp{Also: Why did we test with epsilon=2.58??}\hzh{Because PrivTree's
%      epsilon itself is a pretty complex forumla. If we set all the
%      parameters to that formula with $1.0$, then we get $2.58$.}
%    \bcp{Deserves a comment.}}
\label{fig:evaluation}

\end{figure}
The results are summarized in \Cref{fig:evaluation}. For SparseVector, we
implemented all four of the non-differentially-private variants studied
in \cite{Lyu:2017:USV:3055330.3055331}. For the other algorithms, we implemented
one non-differentially private variant. Among these non-private variants, we try
to mimic typical programming mistakes, such as using a wrong variable with name
similar to the correct one, using incorrect width parameter to Laplace sampling
instruction, off-by-one errors, and using the wrong arithmetic operator.
%\iflater
%\BCP{Seems like we should really do ALL of these mistakes for EACH
%  benchmark, no?  No time to do it now, I guess (right? or would it actually
%  be pretty fast?), but we should definitely do it before the author
%  response period because people are going to ask about this for sure.}
%\hzh{We don't have enough time to do it now... but yes I will do this before response period...}
%\fi
%\BCP{Can we say anything
%about how we did that? Did we try to make subtle or straightforward
%mistakes? Why did we only find one way to get each of them wrong?  (Seems
%like for most programs there should be multiple ways.  E.g., for RNM, the
%wrong example we show is prety obvious; we could have made a more subtle
%mistake by choosing the wrong parameters to Laplace.}
The rest of this section reports in more detail on the
two most interesting benchmarks\ifextended: ReportNoisyMaxWithGap---an algorithm whose
optimal privacy parameter is incorrectly rejected by \fuzzidp---and
PrivTree---an algorithm with two important characteristics that makes it
challenging for automatic checking\fi.

\paragraph*{ReportNoisyMaxWithGap}
\label{paragraph:rnmgap}
\ifextended The original \fi ReportNoisyMax \ifextended algorithm \fi returns the index of the largest
value in an input list, after adding some random noise to each of the original
values.  In addition the returning the index of largest noised value,
ReportNoisyMaxWithGap also releases the numerical gap between the noised max
value and noised runner-up value. Surprisingly, this additional information
does not increase
the privacy cost compared to the original ReportNoisyMax
algorithm~\cite{DBLP:journals/corr/abs-1904-12773}.
We ran the testing framework on a correct implementation of
ReportNoisyMaxWithGap with its optimal $\epsilon = 2.0$. However, this claim was
{\em incorrectly} rejected by the testing framework. We manually inspected the
generated symbolic formula for ReportNoisyMaxWithGap, and realized that its path
condition essentially requires releasing the index of the noisy runner-up
value. This path condition is sufficient but not necessary for proving $(2.0,
0)$-differential privacy for ReportNoisyMaxWithGap. Our manual analysis also
revealed that this path condition should lead to $(4.0, 0)$-differential privacy for
ReportNoisyMaxWithGap. We ran another test on ReportNoisyMaxWithGap with
$\epsilon = 4.0$, and this time the testing framework correctly accepted the
algorithm as private.

\paragraph*{PrivTree}
PrivTree \cite{Zhang:2016:PDP:2882903.2882928} is a differentially private
algorithm for building spatial decomposition trees that approximate occupied
regions of space. We implemented a one-dimensional version of the PrivTree
algorithm over the unit interval.

PrivTree is challenging for automatic checking for at least two reasons. The
first is that PrivTree has an unbounded loop that terminates with
probability $1$---the probability of PrivTree \emph{not} terminating after $n$
iterations of its main loop vanishes as $n$ increases, such that the main loop
eventually surely terminates, but the exact number of iterations cannot be
statically computed. The second is that the privacy analysis for PrivTree
involves intermediate privacy costs that depend on input
values \cite{Zhang:2017:LTA:3093333.3009884}.

The first characteristic poses issues for general static analysis, as we cannot
statically unroll PrivTree's main loop. \fuzzidp's symbolic interpreter is also
susceptible to this challenge. However, for our testing strategy in
\Cref{sec:testing-differential-privacy} to work, the symbolic interpreter
only needs
to produce trees that match those {\em observed} in the instrumented execution. With
each additional iteration of PrivTree, the size of the spatial decomposition
tree strictly increases. Thus, our symbolic interpreter can eagerly cut off the
rest of the infinite search once it realizes that all future iterations will
produce trees that do not match those observed in the instrumented executions

The second characteristic poses issues for tools that
generate {\em proofs} of privacy. As such tools need to
reason over all
possible similar inputs, the input values must be universally quantified and
unknown at proof-generation time. So intermediate privacy costs that depend on
input values are arithmetic expressions over variables representing these unknown inputs. For PrivTree, these intermediate privacy cost expressions involve
non-linear arithmetic, an undecidable theory that can only be solved in a
best-effort way by SMT solvers.

By contrast, our testing framework repeatedly chooses  pairs
of \emph{concrete} inputs to PrivTree. This
means that the intermediate privacy costs
can be represented with much simpler symbolic formula, making automatic
privacy analysis for PrivTree feasible with Z3.

\section{Evaluated Programs}
\label{ap:detailed-evaluation}

\subsection{ReportNoisyMax}
\label{subsec:eval-report-noisy-max}
We presented the source code of ReportNoisyMax in \Cref{rnm}.  The privacy
analysis in \Cref{sec:testing-differential-privacy} shows that this algorithm is
$(2, 0)$-differentially private\iflater, regardless of how long the input list
is\fi.
\iffalse
In practice, the list of inputs to ReportNoisyMax are deterministic query
results executed on some private database, and ReportNoisyMax can discover which
query has the largest value without disclosing the actual value of the query
result. Data analysts may use ReportNoisyMax to save privacy budgets when the
value of these query results are not important to their tasks.
\fi

To test ReportNoisyMax, we need to generate inputs whose coordinate-wise
distance (\Cref{defn:coordinate-wise-distance-relation}) is bounded by $1$. We
implement such a generator manually using
QuickCheck~\cite{Claessen:2000:QLT:351240.351266}, and repeat \fuzzidp's tests
for as long as allowed by the virtual test machine. All of our experiments on
ReportNoisyMax so far have yielded successful results.
%\bcp{I was a little surprised to see QuickCheck coming
%  up here.  Do we {\em always} need to use QC or something like it to
%  generate inputs?  If so, perhaps it should be mentioned earlier?}
%\bcp{Also, are these QC generators generic, or do we need to implement a
%  fresh one for every program we want to test?}
%HZH: We don't have to use QC, but it was convenient.

We also implemented an incorrect version that passes the first value from
the input list
to \lstinline|rnmAux| without adding Laplace noise. The mistake is
highlighted (in bold and red) below:
%\bcp{Looks like the
%  ``do'' is also bolded.  Confusing!  Maybe remove all the non-mistake bold
%  from program listings?  Maybe even highlight it in bold-and-red?}
% \hzh{it's now in bold and red. I have some deuteranopia and it is hard to
% tell that the bold x is actually red... does it look clear for you?}
% BCP: Looks good.  But I wonder whether we should unbold all the haskell
% identifiers elsewhere?
\begin{lstlisting}
rnmBuggy []     =
  error "rnm received empty input"
rnmBuggy (x:xs) = &$\mathtt{do}$&
  xNoised  <- lap x 1.0
  xsNoised <- mapM (\x.lap x 1.0) xs
  rnmAux xsNoised 0 0 &\color{red}$\pmb{x}$&
\end{lstlisting}
\fuzzidp's testing framework catches this mistake very quickly since the
un-noised first input to \lstinline|rnmAux| immediately invalidates
ReportNoisyMax's privacy analysis for many similar input values.

\subsection{SparseVector}
\label{subsec:eval-sparse-vector}
The SparseVector algorithm takes a list of real-valued inputs as the private
input and returns the indices of {\em all} the large elements of the
input. Two input
lists are again considered similar if their coordinate-wise
distance is bounded by $1$.

In addition to the private input list, SparseVector takes two additional,
non-private parameters: $\mathtt{n} :: \mathtt{Int}$ and $\mathtt{threshold}
:: \mathtt{Double}$. SparseVector produces a list of boolean values where
each \lstinline|True| represents an above-\lstinline|threshold| input value. The
parameter \lstinline|n| bounds the total number of \lstinline|True|s
SparseVector may emit before the rest of the computation is truncated. As an
example, if the input is a three-element list and \lstinline|n| is $1$, then
SparseVector may return any of \lstinline|[True]|, \lstinline|[False,True]|,
and \lstinline|[False,False,True]|.

We implement SparseVector in \fuzzidp as follows:
\begin{lstlisting}
sv xs n thresh = do
 let width = 4.0 * fromIntegral n
 thresh' <- lap thresh 2.0
 xs'     <- mapM (\x. lap x width) xs
 svAux xs' n thresh' nil

svAux []     _ _      acc = return acc
svAux (x:xs) n thresh acc
 | n <= 0    = return acc
 | otherwise = do
    let recur n' acc' =
          svAux xs n' thresh acc'
    if (x > thresh)
       (recur (n-1) (snoc acc True))
       (recur n     (snoc acc False))
\end{lstlisting}
Here, \lstinline|nil| %
%\iflater
%\bcp{why not Nil?}\hzh{because it's part of the shallow
%  embedding API, instead of a list data constructor.}\bcp{Do you mean that
%  ``nil'' is a constructor of the Fuzz datatype?  But (if so), then why is it
%  not capitalized when all the other constructors you've shown are
%  capitalized?}\hzh{``nil'' is not a constructor, but it's defined to be equal
%  to the ``Fuzz'' constructor ``Nil''. This arrangement is to forbid
%  programmers from pattern matching on values of the ``Fuzz'' datatype---this
%  datatype is kept abstract in the implementation. We only expose these
%  non-pattern-matchable functions/values for constructing ``Fuzz'' terms.}
%\bcp{Ah.  Probably worth explaining at some point, maybe in a footnote, but
%  not for the submission draft.}
%\fi
%
is an empty list, and \lstinline|snoc| is a function that takes
a \fuzzidp list and a list element, and returns a new list with the element
appended to the end of the supplied one.

\citet{Lyu:2017:USV:3055330.3055331} studied six published variants of SparseVector, among
which only two actually satisfy differential privacy with the intended privacy
parameter. The SparseVector implementation shown here is the a correct version
proposed by Lyu et al., named Algorithm 1
in \cite{Lyu:2017:USV:3055330.3055331}. This variant is $(1, 0)$-differentially
private. \fuzzidp's testing framework correctly rejects the four incorrect
variants and accepts the two correct variants. We re-used the generators for
ReportNoisyMax to generate input lists for SparseVector, and we used QuickCheck
to repeat the testing process for as long as allowed by our virtual test machine.
%until at least 100 test cases passed on each
%correct variant.\bcp{Why did we choose 100?}\hzh{because that's the basically
%largest number of tests we can run for validation, while still fitting an entire
%test suite run into the 6-hour schedule... should we point this out?}
%\bcp{Doesn't seem like a very principled reason, and it further confuses the
%explanation of what our numbers mean: are they representative of 100 tests,
%or of 100 tests repeated many times, or of 600 CPU hours, or...?}\hzh{the
%numbers shown are only for buggy algorithms. the 100 tests are run for
%correct algorithms to make sure our framework does not reject them.}\bcp{I
%know.  But my complaint stands.}

\subsection{PrefixSum}
\label{subsec:eval-prefix-sum}
The PrefixSum algorithm takes a list of numbers and returns a list of the same
length, where each $i$th value in the list is the sum of the values at index $0,
1, \dots, i$. Two inputs to PrefixSum are considered similar if they have the
same length and their $L1$-distance (\Cref{defn:l1-distance}) is bounded by
$1$. The PrefixSum algorithm achieves $(1, 0)$-differential privacy by adding
Laplace noise with width $1$ to each of the values in the input list before
summing. The \fuzzidp code implementing PrefixSum is shown here:
\begin{lstlisting}
ps xs = do
  xs' <- mapM (\x. lap x 1.0) xs
  return (psAux (reverse xs') nil)

psAux []     acc = acc
psAux (x:xs) acc =
  psAux xs (cons (sum (x:xs)) acc)
\end{lstlisting}
The $\mathtt{reverse} :: [a] \rightarrow [a]$ function comes from Haskell's
standard library; it returns a list in the reversed order. (Since the helper
function \lstinline|psAux| accumulates the prefix sums in the reverse order,
we also need to reverse the noised list of input values.)

We use another manually written generator for $L1$-distance-bounded pairs of
lists to test PrefixSum, and we similarly repeat the testing process until
at least 100 test cases pass.

\subsection{SmartSum}
\label{subsec:eval-smart-sum}
The SmartSum algorithm, also known as the Binary Mechanism, is a sophisticated
improvement over PrefixSum that provides the same $(1, 0)$-differential privacy
guarantees, but releases noised sums with much smaller asymptotic
error \cite{Chan:2011:PCR:2043621.2043626}. Data analysts can benefit from the
accuracy improvement with no sacrifice in privacy guarantees by using SmartSum
instead of PrefixSum.

The core idea behind SmartSum is to build a binary tree of partial sums from the
input values, instead of summing up each prefix. We defer the code listing of
SmartSum to \Cref{ap:smart-sum}, where we also show a buggy variant that we
created unintentionally along the way. We reuse the generator for PrefixSum
to test SmartSum. \fuzzidp's testing framework accepts the correct
implementation and rejects our incorrect one.

\subsection{ReportNoisyMaxWithGap}
\label{subsec:eval-report-noisy-max-gap}
\citet{DBLP:journals/corr/abs-1904-12773} recently proposed novel variants
of the ReportNoisyMax
and SparseVector algorithms that release more information about the input data
without increasing their privacy cost.

For ReportNoisyMaxWithGap, the extra information released is the
numerical gap between the largest noised value and the second largest noised
value. The implementation is very similar to ReportNoisyMax:
\begin{lstlisting}
rnmGap []  =
  error "rnmGap received empty input"
rnmGap [_] =
  error "rnmGap received only one input"
rnmGap (x:y:xs) = do
  x' <- lap x 1.0
  y' <- lap y 1.0
  xs' <- mapM (\x. lap x 1.0) xs
  if (x' > y')
    (rnmGapAux xs' 1 0 x' y')
    (rnmGapAux xs' 1 1 y' x')

rnmGapAux [] _ maxIdx
          currMax currRunnerUp =
  return (maxIdx, currMax - currRunnerUp)
rnmGapAux (x:xs) lastIdx maxIdx
          currMax runnerUp = do
  let thisIdx = lastIdx + 1
  let recur = rnmGapAux xs thisIdx
  if (x > currMax)
    (recur thisIdx x currMax)
    (if (x > runnerUp)
      (recur maxIdx currMax x)
      (recur maxIdx currMax runnerUp))
\end{lstlisting}
The algorithm keeps track of the runner-up at each iteration in addition to
the current maximum value and its index in the input list, and eventually
returns the index of the largest value, and the difference between the maximum
and the runner-up. We reuse the generator for ReportNoisyMax to test
ReportNoisyMaxWithGap with the same privacy parameter $\epsilon =
2.0$.

However, \fuzzidp's testing framework incorrectly rejects this claim. We
manually inspected the generated symbolic formula to investigate the cause of
rejection. The path conditions \fuzzidp collects effectively requires the
algorithm to also release the index of the runner-up value, and this is
equivalent to running ReportNoisyMax twice, once to find out the index of the
largest value, and a second time to find out the largest in the remaining
values. Running ReportNoisyMax twice in such a way induces a privacy parameter
$\epsilon = 2.0 * 2 = 4.0$. We check this conjecture by testing
ReportNoisyMaxWithGap again with $\epsilon = 4.0$, and \fuzzidp indeed accepts
this claim.

This is a case where \fuzzidp fails to accept the optimal privacy parameter. The
authors of \cite{DBLP:journals/corr/abs-1904-12773} use an analysis that does
not require releasing the index of the second largest noised value, but \fuzzidp
can only make use of facts collected from path conditions. In this case, the
path conditions are overly strict, which doubles the privacy bound required
by \fuzzidp to check the differential privacy property.

\subsection{SparseVectorWithGap}
\label{subsec:eval-sparse-vector-gap}
SparseVectorWithGap \cite{DBLP:journals/corr/abs-1904-12773} is an improvement
over SparseVector that releases the numeric gap between noised input values and
the noised threshold, when the noised input value is above the noised
threshold. We implement SparseVectorWithGap using the following \fuzzidp program:
\begin{lstlisting}
svGap xs n thresh = do
  thresh' <- lap thresh 2.0
  let width = 4.0 * fromIntegral n
  xs' <- mapM (\x. lap x width) xs
  svGapAux xs' n thresh' nil

svGapAux []     _n _thresh acc =
  return acc
svGapAux (x:xs)  n  thresh acc
  | n <= 0 = return acc
  | otherwise = do
    let recur n' acc' = svGapAux xs n' thresh acc'
    if (x > thresh)
       (do let acc' =
            snoc acc (just (x - thresh))
           recur (n-1) acc')
       (recur n (snoc acc nothing))
\end{lstlisting}
Compared to SparseVector, instead of returning a list of boolean values,
SparseVectorWithGap returns a list of optional values. Each \lstinline|nothing|
value represents the absence of a value, and each \lstinline|just x| represents
the existence and the value of \lstinline|x|.
%\bcp{Why is ``just'' not
%  capitalized?}\hzh{similar to ``nil'', this is part of the shallow embedding
%  API rather than a data constructor}
Our implementation of
SparseVectorWithGap yields \lstinline|just gap| instead of \lstinline|True| for
each above-threshold noised value and its gap between the noised threshold; it
uses \lstinline|nothing| instead of \lstinline|False| for below-theshold noised
values. The parameter \lstinline|n| again bounds the number of above-threshold
optionals SparseVectorWithGap can emit before the rest of the computation is cut
short.

Unlike ReportNoisyMaxWithGap, SparseVectorWithGap's differential privacy
property is accepted by \fuzzidp's testing framework. As the path conditions
collected in a symbolic execution of \lstinline|svGap| are the same as that
those
collected in a symbolic execution of \lstinline|sv|.  \fuzzidp's testing
framework has no issue accepting SparseVectorWithGap's privacy claims.

\subsection{PrivTree}
\label{subsec:eval-priv-tree}
PrivTree \cite{Zhang:2016:PDP:2882903.2882928} is a differentially private
algorithm for building spatial decomposition trees that approximate occupied
regions of space. We implemented a one-dimensional version of the PrivTree
algorithm over the unit interval. The input is a list of points on the unit
interval, represented by a list $\mathtt{xs} :: [\mathtt{Double}]$.  Two input
lists are similar if their database distance is bounded by $1$
(\Cref{defn:database-distance}). Our implementation of PrivTree outputs a
distribution over spatial decomposition trees over the unit interval,
represented by the type
$\mathtt{Map}\, \mathtt{Node}\, \mathtt{()}$. A \lstinline|Node| in the tree is
an interval \lstinline|Node = (Double, Double)|, representing a sub-interval of
the unit interval. In our implementation of PrivTree, the final output spatial
decomposition tree is a collection of leaf nodes in the tree, and the internal
nodes of the decomposition tree are implicitly represented by their constituent
sub-intervals.
%\bcp{not following this; but I think we are going to end up putting all
%  this text in an appendix anyway}.

For example, if the input list is $[0.1, 0.3]$, then PrivTree may output a
tree with leaf nodes $(0.0, 0.25)$, $(0.25, 0.5)$ and $(0.5, 1.0)$. The first two
leaf nodes are occupied by the input points, while the last leaf node is not
occupied; is created when we split the root node (the unit interval) into
$(0.0, 0.5)$ and $(0.5, 1.0)$.

A naive attempt at building such spatial decomposition trees is to take the
textbook QuadTree algorithm \cite{Finkel:1974:QTD:2697709.2697865} and use
Laplace noise to turn it into a differentially private algorithm. This naive
approach would maintain a queue of spatial sub-regions to analyze. On each
iteration, the algorithm would use the Laplace distribution to obtain a
noisy count
of nodes in the sub-region, then decide whether to split the sub-region by
comparing the noisy count with a pre-determined threshold
value. \citet{Zhang:2016:PDP:2882903.2882928} argue that this method has two
significant drawbacks: 1) to ensure the final privacy parameter has a finite
bound, we must also bound the maximum depth of the spatial tree built by this
procedure, and 2) it is difficult to pick a threshold value that leads to
accurate spatial decomposition trees.

The PrivTree algorithm solves these issues by removing requirements
of both the depth bound and the pre-determined threshold. To save space, we
defer the implementation of PrivTree to \Cref{ap:priv-tree}.

PrivTree is a challenging algorithm for automatic verification for at least
two reasons. The first is that PrivTree terminates
probabilistically, i.e., the probability of PrivTree \emph{not} terminating
after $n$ iterations of its main loop diminishes as $n$ increases. The second
one is that the privacy cost analysis used in the PrivTree's privacy
proof involves intermediate privacy costs that depend on input
values \cite{Zhang:2017:LTA:3093333.3009884}.

The first characteristic poses issues for static analysis, as we cannot
statically know how many iterations PrivTree will run. \fuzzidp's symbolic
interpreter is also susceptible to this challenge. However, recall
from \Cref{sec:testing-differential-privacy} that the symbolic interpreter only needs to produce trees that match those observed in the instrumented
execution. PrivTree would only need to run more iterations if it decides to
split the current sub-region, which means the final tree will contain more and
more leaf nodes as the number of iterations increases. Thus, our symbolic
interpreter can eagerly cut off the rest of the (infinite) search once it
realizes that all future iterations will produces trees that do not match those
observed in the instrumented executions\footnote{This ad-hoc early cut-off is not
hardcoded in the symbolic interpreter, but represented by \lstinline|abort|
instructions in the source code. We will discuss how to generalize the ad-hoc
cut-off in \Cref{sec:limitation}.}, saving \fuzzidp's testing framework from
infinite unrolling.

The second characteristic poses issues for tools aimed at automatically
generating proofs of differential privacy. As such tools need to reason over all
possible input values, these input values must be universally quantified and
unknown at proof-generation time, which means intermediate privacy costs that
depend on input values must be represented by expressions over the unknown input
values. For PrivTree, these intermediate privacy cost expressions involve
non-linear arithmetic, an undecidable theory that can only be solved in a
best-effort way by SMT solvers.

On the other hand, \fuzzidp's testing framework chooses a pair
of \emph{concrete} input values and evaluates PrivTree over these inputs. This
allows the testing framework to represent these intermediate privacy
costs through the much simpler generic shift relation introduced
in \Cref{sec:testing-differential-privacy}.

\fuzzidp's testing framework accepts the correct implementation of PrivTree over
the unit interval. We also implemented a buggy version similar to the naive
approach but not placing a depth bound on the tree created. \fuzzidp's
testing framework correctly rejects this buggy variation.

\subsection{NoisySum, NoisyCount and NoisyMean}
\label{subsec:three-simple-mechanisms}
The NoisySum, NoisyCount, and NoisyMean algorithms are simple mechanisms that
aggregate private data. We show their implementations below.
\begin{lstlisting}
noisySum :: [Expr Double]
         -> Expr (&$\distr$& Double)
noisySum xs = lap (sum xs) 1.0

noisyCount :: [Expr Double]
           -> Expr Double
           -> Expr (&$\distr$& Double)
noisyCount xs threshold = do
  let c = length (filter (>= threshold) xs)
  lap (fromIntegral c) 1.0
\end{lstlisting}
The NoisySum algorithm's similarity relation bounds two lists' $L1$ distance
with $1.0$, and it is $1.0$-differentially private.

The NoisyCount algorithm's similarity relation bounds two lists' database
distance bounded with $1$, while the parameter $\mathtt{threshold}$ is public
data. This algorithm counts the number of values above the given threshold, and
adds noise to the count before releasing it. The NoisyCount algorithm is also
$1.0$-differentially private.
\begin{lstlisting}
noisyMean :: [Expr Double]
          -> Expr Double
          -> Expr (&$\distr$& Double)
noisyMean xs clipBound
  | clipBound < 0 = error "simpleMean: received clipBound < 0"
  | otherwise = do
      s <- clippedSum xs 0 clipBound
      noisedS <- lap s 1.0
      let count =
        fromIntegral (lit (length xs))
      noisedC <- lap count 1.0
      return (noisedS, noisedC)

clippedSum []     acc clipBound =
  return acc
clippedSum (x:xs) acc clipBound =
  ifM (x >= clipBound)
    (clippedSum xs (acc + clipBound))
    (ifM (x < (-clipBound))
       (clippedSum xs (acc - clipBound))
       (clippedSum xs (acc + x)))
\end{lstlisting}
The NoisyMean algorithm's similarity relation also bounds two lists' database
distance bounded with $1$, while the parameter $\mathtt{clipBound}$ is public
data. NoisyMean is $(\mathtt{clipBound}+1.0)$-differentially private. A critical
intermediate step in NoisyMean is to clip each input value into the range
$[-\mathtt{clipBound}, \mathtt{clipBound}]$ before summing. This step is
necessary because extreme outliers will lead to violations of differential
privacy.

\newpage
\onecolumn
\section{Implementation of $\mathtt{eval}$}
\label{ap:eval-implementation}
%\bcp{Most of this could be elided --- it's just useful to see the monadic
%  onces, If (where is that, by the way??), Laplace (what is {\tt Laplace}
%  here??), and a few more.}
\begin{figure}[H]
\begin{lstlisting}[xleftmargin=.10\textwidth]
return  :: a -> &$\distr$& a
(>>=)   :: &$\distr$& a -> (a -> &$\distr$& b) -> &$\distr$& b
laplace :: Double -> Double -> &$\distr$& Double

eval :: Expr a -> a
eval (Literal a)   = a
eval (Return a)    = return (eval a)
eval (Bind a f)    = eval a >>= (eval . f . Literal)
eval (Laplace c w) = laplace (eval c) w
eval (If cond a b) = if (eval cond) (eval a) (eval b)
eval (Add a b)     = (eval a) + (eval b)
eval (Lt a b)      = (eval a) < (eval b)
eval (Loop acc pred iter) = eval (Loop acc pred iter) =
  runLoop (eval acc) (eval . pred . Lit) (eval . iter . Lit)

runLoop :: Monad m =>
  a -> (a -> bool) -> (a -> m a) -> m a
runLoop acc pred iter = do
  if pred acc
  then do
    acc' <- iter acc
    runLoop acc' pred iter
  else
    return acc
\end{lstlisting}
\caption{Source code for $\mathtt{eval}$}
\label{code:eval}
\end{figure}

\section{SmartSum}
\label{ap:smart-sum}
\begin{figure}[H]
\begin{lstlisting}[xleftmargin=.15\textwidth]
smartSum xs = smartSumAux xs 0 0 0 0 nil

smartSumAux []     _    _ _ _   results = return results
smartSumAux (x:xs) next n i sum results = do
  let sum' = sum + x
  if_ ((mod (i + 1) 2) == 0)
      (do n' <- lap (n + sum') 1.0
          smartSumAux xs n'    n' (i+1) 0    (snoc results n'))
      (do next' <- lap (next + x) 1.0
          smartSumAux xs next' n  (i+1) sum' (snoc results next'))
\end{lstlisting}
\caption{Source code for SmartSum}
\label{code:smart-sum}
\end{figure}

\begin{figure}[H]
\begin{lstlisting}[xleftmargin=.15\textwidth]
smartSumBuggy xs = smartSumAuxBuggy xs 0 0 0 0 nil

smartSumAuxBuggy []     _    _ _ _   results = return results
smartSumAuxBuggy (x:xs) next n i sum results = do
  let sum' = sum + x
  if_ ((mod (i + 1) 2) == 0)
      (do n' <- lap (n + sum') 1.0
          smartSumAuxBuggy xs n'    n' (i+1) sum' (snoc results n'))    -- bug
      (do next' <- lap (next + x) 1.0
          smartSumAuxBuggy xs next' n  (i+1) sum' (snoc results next'))
\end{lstlisting}
\caption{Source code for SmartSumBuggy}
\label{code:smart-sum-buggy}
\end{figure}

\section{PrivTree}
\label{ap:priv-tree}
\begin{figure}[H]
\begin{lstlisting}[xleftmargin=.15\textwidth]
privTree xs = privTreeAux xs [rootNode] (S.singleton rootNode) emptyTree

privTreeAux points queue leafNodes tree
  | length leafNodes > k_PT_MAX_LEAF_NODES
  -- to avoid infinite unrolling in symbolic execution
  = abort "unreachable code: there are too many leaf nodes"
  | otherwise
  = case queue of
      [] -> return tree
      (thisNode:more) -> do
        let biasedCount =
          countPoints points thisNode - depth thisNode * k_PT_DELTA
        biasedCount' <-
          if (biasedCount > (k_PT_THRESHOLD - k_PT_DELTA))
             (return biasedCount)
             (return $ k_PT_THRESHOLD - k_PT_DELTA)
        noisedBiasedCount1 <- lap biasedCount' k_PT_LAMBDA
        let updatedTree = updatePT thisNode () tree
        if (noisedBiasedCount1 > k_PT_THRESHOLD)
           (do let (left, right) = split thisNode
               let leafNodes' =
                     S.insert right
                              (S.insert left (S.delete thisNode leafNodes))
               privTreeAux points (more++[left,right]) leafNodes' updatedTree
           )
           (privTreeAux points more leafNodes updatedTree)
\end{lstlisting}
\caption{Source code for PrivTree}
\label{code:priv-tree}
\end{figure}

In the implementation of PrivTree, we keep track of the current set
of \lstinline|leafNodes|, and cut off the rest of the computation when there are
more leaf nodes than we have observed in the instrumented execution.

\section{Extracted Python3 code}
\label{ap:extracted-code}

\begin{lstlisting}
def loop_geometric(true_answer_sens_eps):
  """
  :param true_answer_sens_eps: an array-like of (true_answer, (sensitivity, epsilon))
  :return: a list of (noised_answer, variance)
  """
  loop_acc = (true_answer_sens_eps, [])
  loop_cond = not [] == loop_acc[0]
  while loop_cond:
    if loop_acc[0]:
      uncons_head = loop_acc[0][0]
      uncons_tail = (loop_acc[0])[1:]
      uncons_result = (uncons_head, uncons_tail)
    else:
      uncons_result = None
    if loop_acc[0]:
      uncons_head1 = loop_acc[0][0]
      uncons_tail1 = (loop_acc[0])[1:]
      uncons_result1 = (uncons_head1, uncons_tail1)
    else:
      uncons_result1 = None
    if loop_acc[0]:
      uncons_head2 = loop_acc[0][0]
      uncons_tail2 = (loop_acc[0])[1:]
      uncons_result2 = (uncons_head2, uncons_tail2)
    else:
      uncons_result2 = None
    x = (prim_symmetric_geometric(uncons_result[0][0], (np.exp((0.0 - uncons_result1[0][1][1]) / uncons_result2[0][1][0]))))
    if loop_acc[0]:
      uncons_head3 = loop_acc[0][0]
      uncons_tail3 = (loop_acc[0])[1:]
      uncons_result3 = (uncons_head3, uncons_tail3)
    else:
      uncons_result3 = None
    if loop_acc[0]:
      uncons_head4 = loop_acc[0][0]
      uncons_tail4 = (loop_acc[0])[1:]
      uncons_result4 = (uncons_head4, uncons_tail4)
    else:
      uncons_result4 = None
    if loop_acc[0]:
      uncons_head5 = loop_acc[0][0]
      uncons_tail5 = (loop_acc[0])[1:]
      uncons_result5 = (uncons_head5, uncons_tail5)
    else:
      uncons_result5 = None
    if loop_acc[0]:
      uncons_head6 = loop_acc[0][0]
      uncons_tail6 = (loop_acc[0])[1:]
      uncons_result6 = (uncons_head6, uncons_tail6)
    else:
      uncons_result6 = None
    if loop_acc[0]:
      uncons_head7 = loop_acc[0][0]
      uncons_tail7 = (loop_acc[0])[1:]
      uncons_result7 = (uncons_head7, uncons_tail7)
    else:
      uncons_result7 = None
    if loop_acc[0]:
      uncons_head8 = loop_acc[0][0]
      uncons_tail8 = (loop_acc[0])[1:]
      uncons_result8 = (uncons_head8, uncons_tail8)
    else:
      uncons_result8 = None
    if loop_acc[0]:
      uncons_head9 = loop_acc[0][0]
      uncons_tail9 = (loop_acc[0])[1:]
      uncons_result9 = (uncons_head9, uncons_tail9)
    else:
      uncons_result9 = None
    loop_acc = (uncons_result3[1], loop_acc[1] + [(x, 2.0 * (np.exp((0.0 - uncons_result4[0][1][1]) / uncons_result5[0][1][0])) / ((1.0 - (np.exp((0.0 - uncons_result6[0][1][1]) / uncons_result7[0][1][0]))) * (1.0 - (np.exp((0.0 - uncons_result8[0][1][1]) / uncons_result9[0][1][0])))))])
    loop_cond = not [] == loop_acc[0]
  x1 = loop_acc
  return x1[1]
\end{lstlisting}
\end{document}
\fi % ifappendix